\documentclass[a4paper,onecolumn,superscriptaddress,11pt]{quantumarticle}
\pdfoutput=1
\usepackage[utf8]{inputenc}
\usepackage[english]{babel}
\usepackage{amsmath}
\usepackage{amssymb}
\usepackage{amsthm}
\usepackage[destlabel]{hyperref}
\usepackage{csquotes}
\usepackage[style=alphabetic]{biblatex}
\usepackage{subcaption}
\usepackage{multirow}
\usepackage{nicefrac}
\usepackage{algorithm}
\usepackage[noend]{algpseudocode}
\usepackage{thmtools} 
\usepackage{thm-restate}
\usepackage[nameinlink, capitalise]{cleveref}
\algnewcommand\Input{\textbf{Input: }}
\algnewcommand\Output{\textbf{Output: }}

\usepackage{tocbasic}
\DeclareTOCStyleEntry[
  beforeskip=.5em plus 1pt,
  pagenumberformat=\textbf
]{tocline}{section}

\usepackage{tikz}
\usepackage{enumitem}
\setlist{nosep} 

\newtheorem{theorem}{Theorem}
\counterwithin{figure}{section}
\counterwithin{table}{section}
\counterwithin{theorem}{section}

\newtheorem{proposition}[theorem]{Proposition}
\newtheorem{lemma}[theorem]{Lemma}
\newtheorem{corollary}[theorem]{Corollary}

\newtheorem{definition}[theorem]{Definition}
\newtheorem{problem}[theorem]{Problem}

\theoremstyle{remark}

\newcommand*{\q}{\mathbb{Q}}

\newcommand{\ket}[1]{|#1\rangle}

\newcommand{\U}{\mathrm{U}}

\renewcommand{\O}{\mathcal{O}}
\renewcommand{\U}{\mathfrak{U}}
\newcommand{\Oe}{\mathcal{O}_E}
\newcommand{\R}{\mathcal{R}}
\newcommand{\p}{\mathfrak{p}}
\newcommand{\n}{\mathfrak{n}}
\newcommand{\s}{\mathfrak{s}}

\renewcommand{\a}{\mathfrak{a}}
\newcommand{\nuv}{\Vec{\nu}}
\newcommand{\Bi}{B_{\mathrm{in}}}
\newcommand{\Bo}{B_{\mathrm{out}}}
\newcommand{\inn}{{\mathrm{inn}}}
\newcommand{\out}{{\mathrm{out}}}
\newcommand{\Op}{\mathcal{O}_\p}
\newcommand{\wmaj}{{\preceq_w}\,}
\newcommand{\lmaj}{{\prec_w}\,}
\newcommand{\bv}{\mathbf{v}}
\newcommand{\lU}{\U}
\newcommand{\lD}{\mathfrak{D}}

\title{Multi-qubit circuit synthesis and Hermitian lattices}
\author{Vadym Kliuchnikov, Sebastian Sch{\"o}nnenbeck }
\date{May 2024}

\begin{document}

\maketitle

\begin{abstract}
We present new optimal and heuristic algorithms for exact synthesis of multi-qubit unitaries and isometries.
For example, our algorithms find Clifford and T circuits for unitaries with entries in $\mathbb{Z}[i,1/\sqrt{2}]$.
The optimal algorithms are the A* search instantiated with a new data structure for graph vertices and new consistent heuristic functions. 
We also prove that for some gate sets, best-first search synthesis relying on the same heuristic is efficient.
For example, for two-qubit Clifford and T circuits, our best-first search runtime is proportional to the T-count of the unitary. 
Our algorithms rely on Hermite and Smith Normal Forms of matrices with entries in a ring of integers of a number field,
and we leverage the theory of and algorithms for Hermitian lattices over number fields to prove efficiency. 
These new techniques are of independent interest for future work on multi-qubit exact circuit synthesis and related questions.
\end{abstract}

\newpage

\setcounter{tocdepth}{3}
\tableofcontents

\newpage

\section{Introduction}

Our goal is provide new algorithms for exact synthesis problems. 
One of the most well-studied examples of an exact synthesis problem is the following.
Given an isometry $U$ from $n'$ to $n$-qubits with entries in $\mathbb{Z}[i,1/\sqrt{2}]$,
find an $n$-qubit Clifford and T circuit that implements this isometry. 
More formally, find a sequence of $n$-qubit Clifford unitaries $C_1,\ldots,C_{m+1}$ and integers $j_1,\ldots,j_m \in [n]$
such that 
$$
 U = e^{i \phi} \left( C_1 T_{j_1 } C_2 T_{j_2} \ldots C_{m} T_{j_m} C_{m+1} \right) I_{2^n \times 2^{n'}},
$$
where $T_j$ is a $T$ gate acting on qubit $j$ and 
\begin{equation}
\label{eq:id-isometry}
    I_{2^n \times 2^{n'}} 
    = \left( \begin{array}{c}
         I_{2^{n'} \times 2^{n'}} \\
         \hline
         \textbf{0}_{(2^n -2^{n'}) \times 2^{n'}  } 
    \end{array}\right) 
\end{equation}
When $n'=n$, isometry $U$ is a unitary and above problem is the exact unitary synthesis problem.
When $n'=0$, isometry $U$ is a state and above problem is the exact state preparation problem.

It is also common to assign cost to Clifford and T circuits and consider a problem of finding the cheapest 
circuit that implements given isometry. Typically, all Clifford unitaries cost zero and T gates cost one.
Cost of the circuit is the sum of the costs of the gates that form the circuit.
We will say that a synthesis algorithm is optimal if it finds a cheapest circuit.

We consider several directions of the generalizations of the above example of an exact synthesis problem. 
The first direction is to use isometries with entries in different rings, such as $\mathbb{Z}[i,1/2]$,
$\mathbb{Z}[\sqrt{2},1/2]$ and $\mathbb{Z}[e^{2i\pi/3},1/3]$.
The second directions is to replace the Clifford group with its real sub-group, the real Clifford group, or qutrit Clifford group (using qutrits instead of qubits).
The third direction is to use different non-Clifford gates with different costs instead of the T gate.
In all the generalizations we assign cost zero to the real, qutrit or usual Clifford unitaries.
Our main contribution is to provide a common optimal and heuristic exact synthesis algorithm for these problems.

We formulate exact synthesis problem as a graph search problem.
We show that for a wide range of exact synthesis problems we can use informed search,
such as A* search to find an optimal solution. 
A* search is an algorithm for finding a path from a target vertex of a weighted graph to some fixed source vertex.
A* search relies on function $h(v)$, called heuristic, that provides an estimate of the cost 
of the path between vertex $v$ and the source vertex to improve search efficiency in comparison to an uninformed search, 
such as breadth-first or depth-first search.
When function $h$ satisfies certain properties, it is called a consistent heuristic.
One of our main technical contributions is several constructions of consistent heuristic functions for 
exact synthesis problems.
A* search with consistent heuristic finds an optimal solution, and is also provably the most efficient search algorithm 
amongst all the search algorithms that rely on the same consistent heuristic.
The construction of the heuristic functions relies on relatively simple number-theoretic methods, 
such as Hermite and Smith Normal Forms (\cref{sec:hermite-normal-form}, \cref{sec:snf}).

To instantiate A* search and other similar informed graph search algorithm we need to perform the following steps:

\begin{itemize}
    \item Define a weighted graph, its vertices, edges and edge weights~(\cref{sec:problem-graph})
    \item Define a data-structure for graph vertices~(\cref{sec:graph-vertices})
    \item Define a procedure for computing neighbours of a given vertex~(\cref{sec:vertex-neighbors})
    \item Define a heuristic function $h$ and prove that $h$ is consistent~(\cref{sec:advanced-heuristics})
\end{itemize}

Framing exact synthesis problem as an informed graph search problem opens up opportunities to explore the use of many 
informed search algorithms. Informed search algorithms are a vast research area. For example see Ref.~\cite{Edelkamp2010}, an almost nine-hundred-page book.

We review and use the mathematical and computational tools from the theory of Hermitian lattices in \cref{sec:advanced-preliminaries}, \cref{sec:provable} 
to show that the consistent heuristic functions we construct lead to efficient best-first search \cref{alg:best-first-search} for some of the exact unitary synthesis problems (\cref{thm:best-first-search}).
In practice, we interpolate between A* search and best-first search by replacing the consistent heuristic $h(v)$ with $C h(v)$ for some positive constant $C$.
Our proof of the best-first search efficiency leads to an algorithm for enumeration of all unitaries with given values of the heuristic function, in \cref{sec:lattice-graph},
which might be of the independent interest.
We expect that tools from \cref{sec:advanced-preliminaries}, \cref{sec:provable} will find other application in the study of multi-qubit synthesis 
and other questions related to non-Clifford gates.
We list some of the future work directions and open questions in \cref{sec:conclusion}.

The exact synthesis problem and related questions have been active areas of research in quantum computing since 2012.
Generators of groups of unitary matrices with entries in rings similar to $\mathbb{Z}[i,1/\sqrt{2}]$ have been studied in the context of qubits \cite{Kliuchnikov2013}
\cite{Kliuchnikov2012,Forest2015,Kliuchnikov2015,Parzanchevski2017,Amy2019,Amy2023,Giles2012} and qudits 
\cite{Evra2018,Evra2024,Glaudell2018,Glaudell2022,Glaudell2024,Kalra2023,Kalra2024,Jain2020}. 
In some cases, in addition to generators, complete sets of relations have also been found for these groups 
\cite{Selinger2013,Bian2021,Li2021,Makary2021,Bian2022,Bian2023}.
Multi-qubit and multi-qudit exact synthesis algorithms, both optimal and heuristic, are also an active area of research~\cite{Amy2023b,Glaudell2021}, with many recent results
\cite{Mosca2020,Gheorghiu2021,Gheorghiu2021b,Mukhopadhyay2024,Mukhopadhyay2024b}
relying on the channel representation~\cite{Gosset2014}. 

We introduce new mathematical and computational tools as alternatives to the channel representation and demonstrate their first applications to the exact synthesis problem. 
These tools can also be useful for computing generators and relations of groups of unitary matrices with entries in rings similar to $\mathbb{Z}[i,1/\sqrt{2}]$, 
as well as providing new perspectives on past results.



\section{Preliminaries}

\label{sec:preliminaries}

\subsection{Exact synthesis via graph search}
\label{sec:problem-graph}

In this subsection, we start with a more formal version of an 
exact synthesis problem that introduces some necessary notation. 
We use this notation to define a graph related to this problem, that we call a problem graph,
and show that the exact synthesis problem is equivalent to a graph search problem. 
This subsection is a slight generalization of the graph construction in \cite{Gosset2014}.

\begin{problem}[Exact synthesis]
\label{prob:exact-synthesis}
Let $\R$ be a sub-ring of complex numbers and let $\mathcal{C}$ be a finite group 
of $n$-qudit unitaries with entries in $\R$.
Let $G$ be a finite set on $n$-qudit unitaries with entries in $\R$ that are not in $\mathcal{C}$.
Given an isometry $U$ from $n'$ to $n$ qudits, find a sequence $g_1,\ldots,g_m$ of elements of $G$ 
and sequence $C_1,\ldots,C_{m+1}$ of elements of $\mathcal{C}$ such that 
\begin{equation}
\label{eq:iso}
U = \left( C_1 g_1 C_2 g_2 \ldots C_{m} g_m C_{m+1} \right) \left( \ket{0}^{\otimes (n-n')} \otimes I_{n'} \right),
\end{equation}
where $I_{n'}$ is $n'$-qudit identity and $\ket{0}$ is one-qudit zero state.
\end{problem}
When discussing above problem we refer to group $\mathcal{C}$ as \textbf{the cost-zero group}, to ring $\R$ as \textbf{the base ring}, 
to set $G$ as \textbf{the generators}.
It is also common to define \textbf{a cost function} 
$c$ from $G$ into $\mathbb{R}^+$ and associate cost $\sum_{j=1}^m c(g_k)$
with the solution~\cref{eq:iso} to the exact synthesis problem. 
We call the smallest cost the \textbf{optimal cost} of $U$ and a solutions that achieves the smallest cost an optimal solution.
We require that the cost function $c$ is well-behaved, that is for $g_1, g_2$ from $G$ equal up to left and right multiplication by elements of 
$\mathcal{C}$ the cost is the same. 
We also assume that there are finitely many elements $e^{i \phi}$ in the ring $\R$ and include global phase unitaries $e^{i \phi} I_{n}$ into the group $\mathcal{C}$.

To define the graph corresponding to the problem, we recall that part of \cref{eq:iso} can be simplified: 
$$
 C_1 g_1 C_2 g_2 \ldots C_{m} g_m C_{m+1} = (\tilde C_1 g_1 \tilde C^{-1}_1)(\tilde C_2 g_2 \tilde C^{-1}_2) \ldots (\tilde C_m g_m \tilde C_m^{-1}) \tilde C_{m+1}
$$
for appropriately chosen $\tilde C_j$ from $\mathcal{C}$.
For this reason, we can always adjust the generating set $G$ such that \cref{eq:iso} becomes: 
\begin{equation}
\label{eq:iso2}
U = \left( \tilde g_1 \tilde g_2 \ldots \tilde g_m \tilde C_{m+1} \right) \left( \ket{0}^{\otimes (n-n')} \otimes I_{n'} \right),
\end{equation}
where each $\tilde g_j$ is from the new generating set 
$
\tilde G = \left\{ C g C^{-1} : g \in G, C \in \mathcal{C} \right\}
$. We extend the cost function $c$ to $\tilde G$ via equality $c(C g C^{-1}) = c(g)$.

To compress the number of vertices in the graph corresponding to the exact synthesis problem 
we define an $n'$-qudit finite unitary group $\mathcal{C}^{(n')}$ with entries in $R$ 
\begin{equation}
\label{eq:small-subgroup}
 \mathcal{C}^{(n')} = \left\{ U : I_{n'-n} \otimes U  \in \mathcal{C} \right\},
\end{equation}
that we call \textbf{restricted cost-zero group}.  
For example, $\mathcal{C}^{(n)} = \mathcal{C}$ and  $\mathcal{C}^{(0)}$ are phases $e^{i \phi}$
such that $e^{i \phi} I_n$ is in $\mathcal{C}$.
By definition, $I_{n-n'} \otimes \mathcal{C}^{(n')}$ is a sub-group of $\mathcal{C}$.

Now we can define the infinite un-directed \textbf{problem graph} $\mathcal{G}$ corresponding to the exact synthesis problem as:
\begin{enumerate}
    \item graph vertices are sets
    \begin{equation}
    \label{eq:isometry-vertex}
        v_U = \left\{ U\left( \ket{0}^{\otimes (n-n')} \otimes C' \right) : C' \in \mathcal{C}^{(n')} \right\} \text{ for } U \text{ with entries in }\R
    \end{equation} 
    \item Edges are $(v_U,v_{g U})$ for isometries $U$, $g$ from $G$; the edge weights are $c(g)$.
\end{enumerate}
When there are several vertices $v_C$ for $C$ from $\mathcal{C}$, that is when $|\mathcal{C}|/|\mathcal{C}^{(n')}| > 1$, we introduce additional vertex $s$, weight-zero edges $(s,v_C)$ and call $s$ the source of $\mathcal{G}$.
When there is only one vertex $v_C$, we set $s = v_C$ and call it the source vertex of $\mathcal{G}$.

Every solution to the exact synthesis problem with cost $c$ corresponds to a path of weight $c$ in the graph $\mathcal{G}$ 
between $v_U$ and source $s$. 
The shortest path from $v_U$ to the source vertex $s$ correspond to a cheapest circuit that implements $U$.
There is a simple correspondence between any circuit for isometry $U$ and isometry $C_1 U C_2$ for $C_1 \in \mathcal{C}$
and  $C_2 \in \mathcal{C}^{(n')}$. 
We call isometries $C\left( \ket{0}^{\otimes (n-n')} \otimes I_n \right)$ where $C \in \mathcal{C}$ \textbf{cost-zero isometries}.

To instantiate graph-search algorithms we need a compact data-structure for graph vertices.
One example of such a compact data-structure for Clifford and T exact synthesis was introduced in \cite{Gosset2014} and is based on the channel representation of unitaries.
We will introduces a new, alternative even more compact data-structure for graph vertices of graph $\mathcal{G}$ in \cref{sec:graph-vertices}
that relies on some additional property of the finite group $\mathcal{C}$ discussed in \cref{sec:basis-change-property}.
Next we briefly discuss a simple approach to ensure that degree of vertices of $\mathcal{G}$ is small.

\subsubsection{Normalized gate set}
Generators $\tilde G$ used in \cref{eq:iso2} can be somewhat redundant because it is possible that $\tilde g_1,\tilde g_2$
from $\tilde G$ are related by a cost-zero unitary as $\tilde g_2  = \tilde g_1 C$.
We can get rid of generator $\tilde g_2$ by propagating $C$ to the end of \cref{eq:iso2}.
This motivates the following definition: 
\begin{definition}[Normalized gate set]
\label{def:normalized-gate-set}
A set of unitaries $\hat G$ is a normalized gate set with respect to cost-zero group $\mathcal{C}$
if the following holds:
\begin{itemize}
    \item for any $g$ from $G$, $C g C^{-1} = g' C'$ for $g'$ from $\hat G$ and $C'$ from $\mathcal{C}$,
    \item there are no $g_1 \ne g_2$ in $G$ such that $g_1 = g_2 C$ for $C$ from $\mathcal{C}$
\end{itemize} 
\end{definition}

Pauli exponents $e^{i \pi(I-P)/8}$ introduced in~\cite{Gosset2014} are an example of normalized generators for Clifford and T exact synthesis problems.
Using normalized generators $\hat G$ instead of $\tilde G$ when defining edges reduces vertex degree of the problem graph $\mathcal{G}$.
Later, in \cref{sec:normalized-gate-sets-algo} we describe a simple breadth-first search algorithm to construct a normalized gate set starting from any gate set,
when cost-zero group $\mathcal C$ has additional structure described in the next subsection.
 
\subsection{The basis change property}
\label{sec:basis-change-property}
Our exact synthesis methods apply to exact synthesis \cref{prob:exact-synthesis} for which ring $\R$, finite groups $\mathcal{C}, \mathcal{C}^{(n')}$
have additional number-theoretic structure.
In this sub-section we use $n,n',\mathcal{C}$ as defined in \cref{prob:exact-synthesis} and $\mathcal{C}^{(n')}$ as defined in \cref{eq:small-subgroup}.

Let us first describe the structure of rings $\R$.
\begin{definition}[$\xi$-ring]
\label{def:xi-ring}
Let $\O$ be a principal ideal domain, such that its field of fraction has an involution\footnote{involution is a degree-two field automorphsim} $^\ast$.
Let $\xi$ from $\O$ generate a prime ideal $\p = \xi \O$ such that $\p^d = p  \O$ for some integer prime $p$ and positive integer $d$.
We call ring $\R_\xi = \{ z / \xi^k : z\in  \O, k \in \mathbb{N} \}$ 
a $\xi$\textbf{-ring}. 
When $\O$ is a ring of integers of a number field, $\xi$-ring is a \textbf{global $\xi$-ring}.
\end{definition}
We refer to $\xi$ as \textbf{the denominator of} $\R$, to $\O$ as the \textbf{ring of integers of }$\R$, to $^\ast$ as the \textbf{complex conjugation} of $\R$
and to field of fraction of $\O$ as the \textbf{field of fraction} of $\R$.
For the first part of the paper, we will be interested in global $\xi$-ring where field of fraction $E$ is a CM-field or a totally-real field and 
$\O = \Oe$ is the ring of integers of $E$. 
In this case, $\R$ is an example of a ring of $S$-units of a number field $E$ with $S = \{ \p \}$.
We will need more general $\xi$-rings in \cref{sec:provable}.


A simple example of a global $\xi$-ring $\R$ is $\mathbb{Z}[i,1/2]$.
Indeed, $\mathbb{Z}[i,1/2]$ is a sub-ring of number field
$$
\q(i) = \{ q_1 + i q_2 : q_1,q_2 \in \q \},
$$
the ring of integers of $\q(i)$ is 
$$
\mathbb{Z}[i] = \{ z_1 + i z_2 : z_1,z_2 \in \mathbb{Z} \},
$$
and $1+i$ generates a totally ramified ideal $(1+i)^2 \mathbb{Z}[i]  = 2  \mathbb{Z}[i]$.
Ring of integers $\mathbb{Z}[i]$ is a principal ideal domain because there exist Euclidean algorithm in $\mathbb{Z}[i]$.
Up to a choice of global phases, all Clifford group unitaries can be written as unitaries with entries in $\mathbb{Z}[i,1/2]$.
The other examples of global $\xi$-rings $\R$ are summarized in \cref{tab:s-integers-examples}. 

\begin{table}[t]
    \centering
    \begin{tabular}{|c||c|c|c|c|}
    \hline 
    $\R_\xi$ & $E$ & $\Oe$ & $\xi$ & Degree of $E$\tabularnewline
    \hline 
    \hline 
    $\mathbb{Z}[i,1/2]$ & $\mathbb{Q}(i)$ & $\mathbb{Z}[i]$ & $1+i$ & $2$\tabularnewline
    \hline 
    $\mathbb{Z}[\sqrt{2},1/2]$ & $\mathbb{Q}(\sqrt{2})$ & $\mathbb{Z}[\sqrt{2}]$ & $\sqrt{2}$ & $2$\tabularnewline
    \hline 
    $\mathbb{Z}[i,1/\sqrt{2}]$ & $\mathbb{Q}(\zeta_8)$ & $\mathbb{Z}[\zeta_8]$ & $1+\zeta_8$ & $4$\tabularnewline
    \hline 
    $\mathbb{Z}[\zeta_{16},1/(1+\zeta_{16})]$ & $\mathbb{Q}(\zeta_{16})$ & $\mathbb{Z}[\zeta_{16}]$ & $1+\zeta_{16}$ & $8$\tabularnewline
    \hline 
    \hline 
    $\mathbb{Z}[\zeta_3,1/\sqrt{-3}]$ & $\mathbb{Q}(\zeta_3)$ & $\mathbb{Z}[\zeta_3]$ & $\sqrt{-3}$ & $2$\tabularnewline
    \hline 
    \end{tabular}
    \caption{Examples of global $\xi$-rings $\R_\xi$~(\cref{def:xi-ring}) commonly used to formulate exact synthesis problems,
    $\zeta_n$ is the $n$-th root of unity $e^{2\pi i/n}$. 
    Ideal $\xi \Oe$ is a totally ramified prime ideal,
    $\Oe$ is a principal ideal domain.}
    \label{tab:s-integers-examples} 
\end{table}

We also require that cost-zero group $\mathcal{C}$ and related groups have a simple number-theoretic description.
We denote by $\mathcal{U}_{\R,n}$ the set of all $n$-qudit unitaries with entries in $\xi$-ring $\R$~(\cref{def:xi-ring}), that is 
square $d^n \times d^n$ matrices with entries in $\R$ such that $U^\dagger U$ is the identity matrix ($d$ is the qudit dimension).
Note that $U^\dagger = (U^\ast)^T$ with $U^\ast$ being element-wise complex conjugate of $U$ using the complex-conjugation map $^\ast$ of $\R$.
Similar $\mathcal{U}_{\R,n' \rightarrow n}$ are all isometries from $n'$ qudits to $n$ qudits with entries in $\R$.

We require that there exist an invertible matrix $\Bo$ an an invertible matrix $\Bi$ such the following 
equations hold: 
\begin{equation}
\label{eq:basis-change-property}
\begin{array}{rcl}
 \mathcal{C} & = & \left\{  U \in \mathcal{U}_{\R,n} : \Bo^{-1}~U \Bo \text{ has entries in } \Oe \right\}  \\ 
 \mathcal{C}^{(n')} =  \left\{ U  \in \mathcal{U}_{\R,n'} : I_{n'-n} \otimes U  \in \mathcal{C} \right\} & = & \left\{  U \in \mathcal{U}_{\R,n'} : \Bi^{-1}~U \Bi \text{ has entries in } \Oe \right\}  \\
 \left\{ C(\ket{0}^{\otimes (n-n')} \otimes I_{n'}) : C \in \mathcal{C} \right\} & = &  \left\{  U \in \mathcal{U}_{\R,n' \rightarrow n} : \Bo^{-1}~U \Bi \text{ has entries in } \Oe \right\} 
 \end{array}
\end{equation}
We say that a finite $n$-qudit unitary group  $\mathcal{C}$ has \textbf{the basis change property} when it consist of matrices with entries $\xi$-ring $\R_\xi$~(\cref{def:xi-ring}) and it satisfies above equations~\cref{eq:basis-change-property} for all non-negative integers $n' \le n$.
Matrices $\Bi$ and $\Bo$ are the \textbf{input and output basis change matrices}.
Informally, the basis change property states that cost-zero group, cost-zero restricted group and cost-zero isometries become matrices with entries in $\Oe$ when considered in an appropriately chosen bases.
        
Groups $\mathcal{C}$ that satisfy above properties include $n$-qubit Clifford group, $n$-qubit real Clifford group for certain $n,n'$, and one qutrit Clifford group. 
We use the matrices below to express $\Bo, \Bi$ given in \cref{tab:basis-changes} and summarize some of the recent results from~\cite{k2024}.
\begin{equation}
\label{eq:basis-changes}
B_{\mathbb{C}} = 
\left(
\begin{array}{cc}
\frac{1}{1+i} & 0 \\
\frac{1}{1+i} & 1 \\
\end{array}
\right), \quad
B_{\mathbb{R}} = 
\left(
\begin{array}{cc}
\frac{1}{\sqrt{2}} & 0 \\
\frac{1}{\sqrt{2}} & 1 \\
\end{array}
\right), \quad
B_{(3)} = 
\left(
\begin{array}{ccc}
1/\sqrt{-3} & 0 & 0 \\
1/\sqrt{-3} & 1 & 0 \\
1/\sqrt{-3} & 0 & 1 \\
\end{array}
\right)
\end{equation}
\begin{equation}
\label{eq:rational-basis}
B_{\q,n} =
\left(
\begin{array}{c|c}
     \mathrm{Re}(B_{\mathbb{C}}^{\otimes (n-1)}) & -\mathrm{Im}(B_{\mathbb{C}}^{\otimes (n-1)}) \\
     \hline 
     \mathrm{Im}(B_{\mathbb{C}}^{\otimes (n-1)}) &  \mathrm{Re}(B_{\mathbb{C}}^{\otimes (n-1)}).
\end{array}
\right).
\end{equation}
In recent work~\cite{k2024} we have shown that $n$-qubit Clifford group, generated by phase-adjusted Hadamard 
$\tilde H = \frac{1}{1+i}\left( \begin{array}{cc}
    1 & 1 \\
    1 & -1
\end{array} \right)$, 
S and CNOT gates and appropriately chosen global phases satisfies~\cref{eq:basis-change-property} for a wide range of base rings $R$.
We also numerically checked in~\cite{k2024} that the real sub-group of the $n$-qubit Clifford group, that we call real Clifford group, 
also satisfies~\cref{eq:basis-change-property}  for $n=1,2,3$ with $n'=n, 0$.
Additionally, we numerically checked in~\cite{k2024} that one qutrit Clifford group satisfies~\cref{eq:basis-change-property}
with the base ring in \cref{tab:basis-changes}. Additional examples of the groups and corresponding base-rings that satisfy~\cref{eq:basis-change-property} 
can be found in appendices in~\cite{k2024}. We conjecture that basis change property holds for $n$-qubit real Clifford groups and 
$n$-qutrit Clifford groups. 

\begin{table}[t]
    \centering
\begin{tabular}{|cc|c|c|c|c|}
\hline 
\multicolumn{2}{|c|}{Group} & Qudit & Base & \multirow{2}{*}{$\Bo$} & \multirow{2}{*}{$\Bi$}\tabularnewline
\multicolumn{2}{|c|}{$\mathcal{C}$} & dimension & $\xi$-ring $\R_\xi$ &  & \tabularnewline
\hline 
\hline 
Clifford group & $\mathcal{C}_{\mathbb{C},n}$ & 2 & $\mathbb{Z}[i,1/2]$ & {$B_{\mathbb{C}}^{\otimes n}$}, eq.(\ref{eq:basis-changes}) & $B_{\mathbb{C}}^{\otimes n'}$\tabularnewline
\hline 
Clifford group & $\langle\mathcal{C}_{\mathbb{C},n},\zeta_8 I_{n}\rangle$ & 2 & $\mathbb{Z}[i,1/\sqrt{2}]$ & {$B_{\mathbb{C}}^{\otimes n}$} & $B_{\mathbb{C}}^{\otimes n'}$\tabularnewline
\hline 
Real Clifford  & $\mathcal{C}_{\mathbb{R},n}$ & 2 & $\mathbb{Z}[\sqrt{2},1/2]$ & $B_{\mathbb{R}}^{\otimes n}$, eq.(\ref{eq:basis-changes}) & $B_{\mathbb{R}}^{\otimes n'}$\tabularnewline
group &  &  &  & $n=1,2,3$ & $n'=n,0$ \tabularnewline
\hline 
\hline 
Qutrit Clifford  & $\mathcal{C}_{(3)}$ & 3 & $\mathbb{Z}[\zeta_3,1/\sqrt{-3}]$ & {$B_{(3)}$, $n=1$} & $B_{(3)}$,$n'=1$\tabularnewline
group &  &  &  & eq.(\ref{eq:basis-changes})   & \textbf{$1$$,n'=0$}\tabularnewline
\hline 
\end{tabular}
    \caption{Examples of finite groups that admit a simple number-theoretic description together with related subgroups and isometries via a basis change 
    as described by \cref{eq:basis-change-property}~\cite{k2024}.
    Note that $n$-qubit Clifford group $\mathcal{C}_{\mathbb{C},n}$ only includes a global-phase adjusted Hadamard to ensure that all unitaries in it have entries in $\mathbb{Z}[i,1/2]$.
    Similarly, $\mathcal{C}_{(3)}$ contains only global-phase adjusted qutrit Hadamard so that all unitaries in $\mathcal{C}_{(3)}$ have entries in  $\mathbb{Z}[\zeta_3,1/\sqrt{-3}]$.
    We allow $n'=0$, use convention $B^{\otimes 0} = 1$, and treat scalar $1$ as one by one matrix. 
    For the details of $\R_\xi$ see \cref{tab:s-integers-examples}.
    }
    \label{tab:basis-changes}
\end{table}

\subsubsection{Modular basis change property}
\label{sec:modular-basis-change}

For constructions of advanced heuristic functions in~\cref{sec:advanced-heuristics} we require that input and output basis change matrices have an additional property,
that holds for the basis change matrices related to Clifford groups.

\begin{definition}[Unimodular and $\xi$-modular matrices]\label{def:xi-modular}
Matrix $A$ entries in principal ideal domain $\O$ is $\O$\textbf{-unimodular} when it has an inverse with entries in $\O$.
Matrix $A'$ with entries in $\xi$-ring $\R$ with ring of integer $\O$ is $\xi$\textbf{-modular} if there exist integer $k$ such that $A' = \xi^{k} A$ 
for some $\O$-unimodular matrix $A$.
\end{definition}

If a finite $n$-qudit group has a basis change property with the basis change matrices $B$ such that $B^\dagger B$ is $\xi$-modular, we say that the group has a \textbf{modular basis change property}.
For example, this is the case for qubit Clifford groups and real Clifford groups in \cref{tab:basis-changes}.

\subsection{Common unitary matrices, isometries, states and gate sets}
\label{sec:common-matrices}

Here we briefly recall some common matrices and gate sets used in quantum computing. 
\begin{equation}
\label{eq:common-complex}
\text{Z}=\left(\begin{array}{cc}
1 & 0\\
0 & -1
\end{array}\right),\,\text{X}=\left(\begin{array}{cc}
0 & 1\\
1 & 0
\end{array}\right),\,\text{S}=\left(\begin{array}{cc}
1 & 0\\
0 & i
\end{array}\right),\,\text{T}=\left(\begin{array}{cc}
1 & 0\\
0 & \zeta_{8}
\end{array}\right),\,\sqrt{\text{T}}=\left(\begin{array}{cc}
1 & 0\\
0 & \zeta_{16}
\end{array}\right)
\end{equation}

\begin{equation}
\label{eq:common-real}
\text{S}_{y} = \left(\begin{array}{cc}
\cos(\frac{\pi}{4}) & -\sin(\frac{\pi}{4})\\
\sin(\frac{\pi}{4}) & \cos(\frac{\pi}{4})
\end{array}\right),\,\text{T}_{y}=\left(\begin{array}{cc}
\cos(\frac{\pi}{8}) & -\sin(\frac{\pi}{8})\\
\sin(\frac{\pi}{8}) & \cos(\frac{\pi}{8})
\end{array}\right),\,\text{H}=\frac{1}{\sqrt{2}}\left(\begin{array}{cc}
1 & 1\\
1 & -1
\end{array}\right)
\end{equation}

For any $n$-qubit unitary $\text{U}$ we define $(n+1)$-qubit controlled-$U$
as 

\[
\text{CU}=\frac{I_1+Z}{2}\otimes I_n +\frac{I-Z}{2}\otimes U
\]
For example, we use matrices $\text{CT}$,$\text{CS}$, $\text{CH}$,
$\text{CT}_{y}$,$\text{CS}_{y}$ when defining gate sets. 

We list come common gate sets and their names in \cref{tab:common-gate-sets}.
Clifford and $\text{T}$, $\sqrt{\text{T}}$ is commonly referred to as simply Clifford and $\sqrt{\text{T}}$.
Clifford and $\text{T},\text{T}^{\otimes2},\ldots,\text{T}^{\otimes n}$ is used to find T depth optimal circuits.
When using synthesis algorithm described in this paper we derive a normalized gate set for each of this common gate sets
using the algorithm from \cref{sec:normalized-gate-sets-algo}.
In relation to \cref{prob:exact-synthesis}, the gate set names start with the cost-zero group and then list 
the generators with the non-zero cost.

\begin{table}[ht]
    \centering
\begin{tabular}{|c|c|c|c|c|}
\hline 
Gate set & Basis \textbf{$B$} & Field $E$ & Denominator $\xi$ & Minimal $n$\tabularnewline
\hline 
\hline 
Clifford and $\text{T}$ & $B_{\mathbb{C}}^{\otimes n}$ & $\mathbb{Q}(\zeta_{8})$ & $1+\zeta_{8}$ & $1$\tabularnewline
\hline 
Clifford and $\text{T}$, $\sqrt{\text{T}}$ & $B_{\mathbb{C}}^{\otimes n}$ & $\mathbb{Q}(\zeta_{16})$ & $1+\zeta_{16}$ & $1$\tabularnewline
\hline 
\hline 
Clifford and $\text{CS}$ & $B_{\mathbb{C}}^{\otimes n}$ & $\mathbb{Q}(i)$ & $1+i$ & $2$\tabularnewline
\hline 
Clifford and $\text{CH}$ & $B_{\mathbb{R}}^{\otimes n}$ & $\mathbb{Q}(\sqrt{2})$ & $\sqrt{2}$ & $2$\tabularnewline
\hline 
Clifford and $\text{T},\text{T}\otimes\text{T},\text{CT}$ & $B_{\mathbb{C}}^{\otimes n}$ & $\mathbb{Q}(\zeta_{8})$ & $1+\zeta_{8}$ & $2$\tabularnewline
\hline 
Clifford and $\text{T}_{y},\text{CS}_{y},\text{CT}_{y}$ & $B_{\mathbb{R}}^{\otimes n}$ & $\mathbb{Q}(\cos(\frac{\pi}{8}))$ & $2+2\cos(\frac{\pi}{8})$ & $2$\tabularnewline
\hline 
Clifford and $\text{T},\text{T}^{\otimes2},\ldots,\text{T}^{\otimes n}$ & $B_{\mathbb{C}}^{\otimes n}$ & $\mathbb{Q}(\zeta_{8})$ & $1+\zeta_{8}$ & $2$\tabularnewline
\hline 
\hline 
Clifford and $\text{CS}$, $\text{CCZ}$, $\text{CCS}$ & $B_{\mathbb{C}}^{\otimes n}$ & $\mathbb{Q}(i)$ & $1+i$ & $3$\tabularnewline
\hline 
Clifford and $\text{CS}_{y}$, $\text{CCZ}$, $\text{CCS}_{y}$ & $B_{\mathbb{R}}^{\otimes n}$ & $\mathbb{Q}(\sqrt{2})$ & $\sqrt{2}$ & $3$\tabularnewline
\hline 
\end{tabular}
    \caption{Common $n$-qubit gate sets used in quantum computing. See \cref{eq:basis-changes} for the definition of $B_{\mathbb{C}}$, $B_{\mathbb{R}}$.
    Unitaries from the gate set have entries in a global $\xi$-ring $\R$~(\cref{def:xi-ring}) with field of fractions $E$ and denominator $\xi$. 
    The cost-zero group has a modular basis change property with matrix $B$~(\cref{tab:basis-changes}). 
    We omit word ``real'' in front of the ``Clifford'' in the gate set name, as it is implied by the corresponding field $E$ being totally real.
    When $n$-qubit gate set includes $k$-qubit gate $U$ we pad it with identity $U \otimes I_{n-k}$ to obtain an $n$-qubit unitary.}
    \label{tab:common-gate-sets}
\end{table}

When discussing states, we use the following notation:
$$
|0\rangle=\left(\begin{array}{c}
1\\
0
\end{array}\right),|1\rangle=\left(\begin{array}{c}
0\\
1
\end{array}\right),|+\rangle=\frac{1}{\sqrt{2}}\left(\begin{array}{c}
1\\
1
\end{array}\right),|x,y\rangle=|x\rangle\otimes|y\rangle,\,|x,y,z\rangle=|x\rangle\otimes|y\rangle\otimes|z\rangle
$$
For $n$-qubit diagonal unitary D, corresponding state is $|\text{D}\rangle = \text{D} |+\rangle^{\otimes n}$. 
Computational basis states on $n$-qubits are all possible tensor products of $|0\rangle$, $|1\rangle$.
They are sometimes labelled by $n$-bit integers, for example $|7\rangle = |1,1,1\rangle$.
It is sometimes convenient to specify unitaries by their action on computational basis states. 
For example, SWAP is a unitary that maps $|x,y\rangle \mapsto |y,x\rangle$ for $x,y \in \{0,1\}$.
Similarly, isometry known as the quantum AND gate is $|x,y\rangle \mapsto |x,y xy\rangle$ for $x,y \in \{0,1\}$.
Finally, there is a unitary that corresponds to a permutation $\sigma$ on integers from $0$ to $2^n-1$, 
as $|j\rangle \mapsto |\sigma(j)\rangle$.

\subsection{Hermite Normal Form}
\label{sec:hermite-normal-form}

In this subsection we review definitions, key results and procedure for Hermite Normal Forms of matrices with entries in $\mathbb{Z}$ 
and ring of integers $\Oe$ of number field $E$. 
Hermite Normal Form for integer matrices is widely used in computational number theory
and almost as ubiquitous as extended Euclidean algorithm \cite{Cohen1993}. 
The Hermite Normal Norm is a canonical form of a matrix up to the right multiplication by $\mathbb{Z}$-unimodular\footnote{$\mathbb{Z}$-unimodular matrices are simply called unimodular} matrices, 
that is given $M \times N$ integer matrices $A$, there is always exist a unimodular $M\times M$ integer matrix $B$, such that $B A$ is in Hermite Normal Form.
Moreover, two $M \times N$ integer matrices $A, A'$ have the same Hermite Normal From if and only if and only if there exist 
unimodular $M\times M$ integer matrix $B$ such that $A = B A'$.
There are numerous very efficient algorithms for the Hermite Normal Form and this is still an active area of research \cite{Pernet2010,Fieker2014}.
Implementations of state of the art Hermite Normal Form algorithms are available in many number theory libraries \cite{Flint,Nemo}.

It is possible to extend the notion of Hermite Normal Form to rings of integers $\Oe$ that are principal ideal domains.
To ensure uniqueness of the  Hermite Normal Form in this case, one needs to be able to efficiently compute canonical elements 
of $\Oe$ up to multiplications by units in $\Oe$~\cite{Conti90}.
When computing canonical elements of $\Oe$ up to a unit is not very efficient or $\Oe$ is not a principal ideal domain,
one can also use pseudo-matrices and Hermite Normal Form algorithms for pseudo-matrices~\cite{Cohen2000,Fieker2014}.

\subsection{Smith Normal Form}
\label{sec:snf}

We recall some basic results regarding Smith Normal Form for matrices with entries in principal ideal domains.
\begin{theorem}[Theorem 3.8 in \cite{jacobson2009basic}]
\label{thm:smith-normal-form}
Consider $M \times N$ matrix $A$ with entries in a principal ideal domain $\Oe$.
There exist unimodular matrices $L$, $R$ with entries in $\Oe$ such that $L A R$ is equal to the following ``diagonal'' matrix
$$ D = 
\left(\begin{array}{cccc|c}
d_{1} &  &  & 0 & \\
 & d_{2} &  &  & \mathbf{0}_{r\times(N-r)}\\
 &  & \ddots & & \\
 0 &  &  & d_{r}\\
 \hline
 & \mathbf{0}_{(M-r)\times r} &  &  & \mathbf{0}_{(M-r)\times (N-r)}
\end{array}\right)
$$
with $d_j$ being non-zero elements of $\Oe$ such that $d_{j}$ divides $d_{j+1}$ for $j \in [r-1]$ 
and $r$ equal to rank of $A$.
\end{theorem}
We refer to the diagonal matrix $D$ as \textbf{Smith Normal Form} of $A$
and diagonal elements $d_1,\ldots,d_r$ are called \textbf{invariant factors} of $A$.
Importantly, the invariant factors are uniquely defined up to a unit of $\Oe$,
which is a subject of the following theorem: 
\begin{theorem}[Theorem 3.9 in \cite{jacobson2009basic}]
\label{thm:invariant-factors-via-gcd}
Let $A$ be an $M \times N$ matrix of rank $r$ with entries in  principal ideal domain $\Oe$.
For each $j \le r$, let $\Delta_j$ be a greatest common divisor of all $j$-rowed minors of $A$ (that is determinants of all possible $j \times j$
sub-matrices of $A$). Then any set of invariant factors for $A$ differ by units of $\Oe$ from 
\begin{equation}
\label{eq:deltas}
d_1 = \Delta_1, d_2 = \Delta_2 \Delta_1^{-1}, \ldots, d_r = \Delta_r \Delta_{r-1}^{-1}    
\end{equation}
\end{theorem}
Ref.~\cite{jacobson2009basic} provides a self-contained exposition of the two above theorems.
Notion of Smith Normal Form generalizes beyond principal ideal domain~\cite{Cohen2000}, 
however considering $\Oe$ that are principal ideal domains is sufficient for our applications.
All the rings of integers for number fields in \cref{tab:common-gate-sets} are principal ideal domains.

There are many efficient algorithm for computing Smith Normal Form~\cite{Cohen1993,Cohen2000} and the invariant factors of matrix $A$. 
As we will see later, in our applications invariant factors are powers of a prime element of $\Oe$.
In this case there is a particularly efficient algorithm~\cite{Lubeck2002} for computing invariant factors for $\Oe = \mathbb{Z}$
that can be easily generalized to the case when $\Oe$ is a principal ideal domain.

We conclude this section with a simple observation regarding the invariant factors of the product of two matrices.
\begin{lemma}
\label{lem:invariant-factors-products}
Let $A$, $B$ be $N \times M$ and $M \times N'$ matrices with entries in a principal ideal domain $\Oe$. 
Let $\Delta^A_j$, $\Delta^B_j$, $\Delta^{AB}_j$ be products of the invariant factors of $A,B$ and 
$AB$ as defined in \cref{eq:deltas}.
Then $\Delta^A_j \cdot \Delta^B_j$ divides $\Delta^{AB}_j$.
\end{lemma}
The proof of the above lemma directly follows from applying \cref{thm:smith-normal-form} to $A,B$.
First, we express $A = L_A D_A R_A $,  $B = L_B D_B R_B $.
Second, we observe that $AB$ has the same invariant factors as $D_A R_A L_B D_B$. 
Finally we use \cref{thm:invariant-factors-via-gcd} to establish divisibility of $\Delta^{AB}_j$
by $\Delta^A_j \cdot \Delta^B_j$, by relying on standard properties of a determinant of a matrix.

When invariant factors are powers of a prime element of $\Oe$, condition $\Delta^A_j \cdot \Delta^B_j$ divides $\Delta^{AB}_j$
can be rewritten as an inequality. 
Given prime ideal $\p$ of $\Oe$ we use $v_\p(x)$ for the power of $\p$ 
in the factorization of $x \Oe$ into prime ideals.
We use convention $v_\p(0) = +\infty$.
Function $v_\p(x)$ called a $\p$-adic valuation.
The divisibility conditions turns into inequality 
\begin{equation}
\label{eq:p-adic-inequality}
v_\p(\Delta^A_j) + v_\p(\Delta^B_j) \le v_\p(\Delta^{AB}_j).
\end{equation}

Finally, note that we can extend the Smith Normal Form to matrices $A'$ with entries in a field $E$ 
by considering matrices $A = \alpha A'$ where $\alpha$ from $\Oe$ is chosen so that $A$ has entries in $\Oe$.

\subsection{Majorization}
\label{sec:majorization}

Non-increasing non-negative integer vector is a vector $(x_1,\ldots,x_N)$ with non-negative integer coordinates such that $x_1 \ge x_2 \ge \ldots \ge x_N$.
Given a vector $x$ we use $x^{\downarrow}$ for $x$ sorted in non-decreasing order.

For two non-increasing non-negative vectors $x,y$ of length $N$ we use notation $x \wmaj y$ when 
\begin{equation}
\label{eq:weak-majorization}
\sum_{j=1}^{k} x_{j} \le \sum_{j=1}^{k} y_{j} \text{ for all } k \le N.
\end{equation}
We say that $y$ \textbf{weakly majorizes} $x$. 

For two non-increasing integer vectors $x,y$ of length $N$ we use notation  $x \lmaj y$ and say that $y$ \textbf{strictly weakly majorizes}
$x$ when $x \wmaj y$ and there exist $k_0\le N$
such that 
\begin{equation}
\label{eq:strict-majorization}
\sum_{j=1}^{k_0} x_{j} < \sum_{j=1}^{k} y_{j}
\end{equation}

The two kinds of majorization, weak and strictly weak, are related via the following proposition,
which is  a generalization of implication $u < v , x \le y + (u-v) \implies x < y$ for real numbers $x,y,u,v$. 
\begin{proposition}
\label{prop:maj-relations}
Consider four non-increasing non-negative vectors $x,y,u,v$ such that 
$u \lmaj v$ and $x \wmaj y + (u-v)$, then $x \lmaj y$.
\end{proposition}
\begin{proof}
By the definition of lexicographical majorization, there exist $k_0$ such that 
$$
\sum_{j=1}^{k_0} u_j - v_j < 0  
$$
For vectors $x$ and $y + (u-v)$ we have
$$
\sum_{j=1}^{k_0} x_j \le \sum_{j=1}^{k_0} (y_j + u_j - v_j) < \sum_{j=1}^{k_0} y_j ,~~ \sum_{j=1}^{k'} x_j \le \sum_{j=1}^{k'} y_j + u_j - v_j \le  \sum_{j=1}^{k'} y_j \text{ for all } k' \le N 
$$
which shows that $x \lmaj y$.
\end{proof}

\section{Basic algorithms and data structures for graph search}
\label{sec:basic-algorithms}

In this section we discuss how to instantiate informed graph search algorithms for 
exact synthesis~\cref{prob:exact-synthesis}~(\cref{sec:problem-graph}) in the case when cost-zero group $\mathcal{C}$
has the basis change property~(\cref{sec:basis-change-property}).

\subsection{Representing graph vertices}
\label{sec:graph-vertices}
The goal of this sub-section is to show that Hermite Normal Form (for matrices with entries in $\Oe$ or $\mathbb{Z}$) can be used to compactly represent 
vertices of the problem graph described in~\cref{sec:problem-graph} when cost-zero group associate with 
the exact synthesis problem has the basis change property~(\cref{sec:basis-change-property}).

Let us recall some notation. 
The basis change property implies that base ring $\R$ associated with the exact synthesis problem 
is defined in terms of a ring of integers $\Oe$ of number field $E$ and denominator $\xi$ from $\Oe$~(\cref{def:xi-ring}).
Moreover, there exist matrices $\Bi,\Bo$ with entries in $\R$ that transform 
cost-zero group $\mathcal{C}$, restricted cost-zero group $\mathcal{C}^{(n')}$
and cost-zero isometries to matrices with entries in $\Oe$~(\cref{eq:basis-change-property}).

Let us now establish that all the isometries in the same problem graph vertex $v_U$ are exactly 
the ones related by the right-multiplication by an invertible matrix with entries in $\Oe$ (when the isometries are written in appropriate basis).
Consider an isometry $U$ with entries in $\R$, recall that the corresponding graph vertex is defined~(\cref{eq:isometry-vertex}) as 
$$
 v_U =  \left\{ U C' : C' \in \mathcal{C}^{(n')} \right\}.
$$
Let write isometry $U$ is a different basis as $ \Bo^{-1} U \Bi$.
Note that two isometries from $v_U$ written in this different basis are related by right multiplication by an invertible matrix with entries in $\Oe$, 
indeed 
$$
 \Bo^{-1} U C' \Bi =  \Bo^{-1} U  \Bi \, \Bi^{-1} C' \Bi.
$$
Let us show that the opposite is true. If isometry $U'$ is such that 
\begin{equation}
\label{eq:integer-relation}
\Bo^{-1} U' \Bi =  \Bo^{-1} U \Bi \, A, \text{ for some invertible  }A \text{ with entries in } \Oe,
\end{equation}
then $U'$ is in $v_U$.
Matrix $A' = \Bi A \Bi^{-1}$ is a unitary matrix with entries in $R$ such that $ \Bi^{-1} A \Bi$ has entries in $\Oe$ because:
$$
\langle A'v, A'v \rangle  = \langle U A'v, U A'v \rangle = \langle U' v, U' v \rangle = \langle v, v \rangle.
$$
By the basis change property~(\cref{eq:basis-change-property}), matrix $A'$ must be from $\mathcal{C}^{(n')}$, 
and therefore $U' \in v_U$.

\subsubsection{Using matrices with entries in the ring of integer of a number field}

Next we show that any problem graph vertex $v_U$ can be uniquely represented by a single matrix with entries in $\Oe$.
Consider an isometry $U$ with entries in $R$, and let $\nu(U)$ be a smallest integer such that 
\begin{equation}
\label{eq:tilde-u}
\tilde U = \xi^{\nu(U)} \Bo^{-1} U \Bi    
\end{equation}
has entries in $\Oe$. 
Let $\hat U$ be the unique Hermite Normal Form\footnote{See \cref{sec:hermite-normal-form} for the discussion regarding the unique Hermite Normal Form for matrices with entries in $\Oe$.}
of $\tilde U^t$. 
We use $\tilde U$ to compactly represent vertex $v_U$ of the problem graph.

Representation $\hat U$ of $v_U$ is well-defined, that is vertices $v_U$, $v_{U'}$ are equal if and only if $\tilde U' = \tilde U$.
First, if $v_U = v_{U'}$, then $\nu(U) = \nu(U')$ and $\hat U = \hat U'$ by the uniqueness of the Hermite Normal Form.
Second, when $\hat U = \hat U'$ we have $\nu(U) = \nu(U')$ and \cref{eq:integer-relation} holds,
which implies that $U' \in v_U$ as discussed above. 
Indeed, the equality $\nu(U) = \nu(U')$ for $N \times N'$ matrices $U,U'$
follows from the observation that $\mathrm{det}(\hat U^\ast \, \hat U^t ) = (\xi \xi^\ast)^{\nu(U) \cdot N'}$.
The relation between $\Bo^{-1} U' \Bi$ and $\Bo^{-1} U \Bi$ then follows from the 
fact that $\hat U = B \tilde U^t$ for some invertible matrix $B$ with entries in $\Oe$.

\subsubsection{Using integer matrices}

Next we show how to use Hermite Normal Form for integer matrices to represent graph vertices.
This relies on representing elements of ring of integers $\Oe$ of degree $d$ as $d\times d$
integer matrices. 
Recall that every ring of integers $\Oe$ has an integral basis $1 = b_1,b_2,\ldots,b_d$ of elements of $\Oe$,
such that every element $x$ of $\Oe$ can be represented as $d$-dimensional integer vector $\mathrm{vec}(x)$
\begin{equation}
\label{eq:vec}
x = \sum_{j=1}^{d} b_j \mathrm{vec}(x)_j  
\end{equation}
Using above notation we define an integer \textbf{representation matrix} $M_x$ associated with $x$ as 
\begin{equation}
\label{eq:matrix-representation}
 \mathrm{vec}(x y) = M_x \mathrm{vec}(y)   \text{ for all } y \in \Oe 
\end{equation}
For any $N' \times N $ matrix $A$ with entries in $\Oe$, we define and $(d N') \times (d N)$ matrix $Z(A)$ with entries in $\mathbb{Z}$
by replacing each entry of $A$ with the corresponding representation matrix, 
that is 
\begin{equation}
\label{eq:as-integer-matrix}
 \left(Z(A)\right)_{(j-1) \cdot d + j', (k-1) \cdot d + k'} = \left(M_{A_{j,k}}\right)_{j',k'} \text{ for } j \in [N'], k \in [N], j',k' \in [d] 
\end{equation}
Using the introduced notation, we define the second representation of vertices $v_U$ of the problem graph using the Hermite Normal form of $Z(\tilde U^t)$ (with $\tilde U$ from \cref{eq:tilde-u}) that we denote as $\bar U$.
Next we show that $\bar U$ is a well-defined representation of a vertex $v_U$, similar to $\hat U$.

Let us first show that for any $U'$ from $v_U$ we have $\bar U = \bar U'$. 
To establish this we show that a square matrix $A$ with entries in $\Oe$ is $\Oe$-modular if and only if matrix $Z(A)$
is unimodular.
To establish this fact we need to relate $\mathrm{det} Z(A)$ and $\mathrm{det}(A)$ for square matrices $A$ with entries in $\Oe$.
First, note that matrices $M_x, M_y$ commute for all $x,y$ from $\Oe$ because $M_x M_y = M_{xy}$.
Second, we apply the a formula for the determinant of block matrices consisting of commuting blocks \cite{Silvester2000} to 
get 
\begin{equation}
\label{eq:det}
\mathrm{det}\,  Z(A) = \mathrm{det}( M_{\mathrm{det}(A)} )
\end{equation}
Above implies that $\mathrm{det}\,  Z(A) = \pm 1$ if and only if $\mathrm{det}(A)$ is a unit in $\Oe$.
This is because $\mathrm{det}(A)$ is a unit of $\Oe$ if and only if $A$ is invertible, 
and $u$ for $\Oe$ is a unit if and only if $\mathrm{det}( M_u ) = \pm 1$. 
We conclude that $\bar U = \bar U'$ because $\tilde U^t$ and $\tilde U'^t$
are related via right multiplication by an invertible matrix $A$ with entries in $\Oe$ and so $Z(\tilde U^t)$
and $Z(\tilde U'^t)$ are related via right multiplication by an invertible integer matrix $Z(A)$,
because $Z(AB) = Z(A)Z(B)$ for any matrices $A$, $B$ with entries in $\Oe$.

Next we show that $\bar U' = \bar U$ implies that $U'$ is from $v_U$.
As discussed above, it is sufficient to show that \cref{eq:integer-relation} holds.
This follows from the fact discussed in the next paragraph.

Let $B,B'$ be two $N' \times N$ matrices with entries in $\Oe$ of rank $N' \le N$,
and let $Z(B) = A Z(B')$ for some invertible integer $(d N' )\times (d N')$ matrix $A$,
then there exist an invertible matrix $A'$ with entries in $\Oe$ such that $Z(A')=A$. 
Above statement follows from the special case $N= N'$, 
because we can always consider a full-rank square sub-matrix of $B$ and the corresponding 
full-rank sub-matrix of $B'$.
To show this special case, we extend map $x \mapsto M_x$ from $\Oe$ to full number-field $E$.
Non-zero elements of $E$ map to invertible rational $d\times d$ matrices that form degree $d$ 
commutative algebra over $\q$.
Matrices $Z(B)$ for $N \times N$ matrices $B$ with entries in $E$ form a degree $d N^2 $ algebra over $\q$.
When $B$ and $B'$ are invertible square matrices with entries in $E$, $A = Z(B) (Z(B'))^{-1}$ also belongs to 
the algebra and there exist matrix $A'$ with entries in $E$ such that $Z(A') = A$.
We conclude that $A'$ must have entries in $\Oe$ because entries of $A$ are integers.

\subsection{Computing normalized gate sets}
\label{sec:normalized-gate-sets-algo}

A simple application for the graph vertex representations in \cref{sec:graph-vertices} is computing normalized gate 
sets starting from Clifford + $U_1,\ldots,U_m$.
First, without loss of generality assume that graph vertices $v_{U_j}$ are all distinct ( if this not the case we can remove some of the $U_j$). 
Second, apply breadth-first search with source vertices $v_{U_j}$ and the neighbours of a given vertex $v_U$
given by
$$
 v_{gUg^{-1}} \text{ for } g \text{ generators of the Clifford group}.
$$
The breadth-first search leads to a finite graph with vertices $v_U$ indexed by $U$ set $G$ that is a normalized gate set.
We refer to the gate sets constructed by such algorithm as \textbf{normalized} Clifford and $U_1,\ldots,U_m$.
This applies to any finite cost-zero group $\mathcal{C}$ with a basis-change property~(\cref{sec:basis-change-property}), such as real and qudit Clifford groups.
We will also say that $G$ is \textbf{normalization of} $U_1,\ldots,U_m$ with respect to the cost-zero group $\mathcal{C}$. 

\subsection{A consistent heuristic function}
\label{sec:simple-heuristic}

Informed graph search requires a function, called a heuristic function $h(v)$, that estimate the cost of reaching the source vertex from given vertex $v$.
Given a graph with source vertex $s$, heuristic function $h$ is a \textbf{consistent heuristic} when for any vertex $v$, its neighbour $v'$ and edge weight $w(v,v')$
the following holds: 
$$
h(v) \le h(v') + w(v,v'),~~h(s) = 0.
$$
For a problem graph $\mathcal{G}$ associated with an exact synthesis problems as in \cref{sec:problem-graph}, 
the first condition becomes: 
$$
h(v_U) \le h(v,v_{gU}) + c(g), \text{ for all } g \in G,
$$
where $G$ is the set of generators and $c(g)$ is the generator cost.
Next we show that when cost-zero group $\mathcal{C}$ in the exact synthesis problem has basis-change property~(\cref{sec:basis-change-property}), 
$\nu(U)$ introduced in \cref{sec:graph-vertices}, \cref{eq:tilde-u} is a consistent heuristic if $c(g) = \nu(g)$.

Let $\Bi, \Bo$ be input and output basis change matrices associated with the basis change property of cost-zero group $\mathcal{C}$,
and let $\Oe$ be the ring of integers and let $\xi$ be the denominator associated with the base ring $\R$ of $\mathcal{C}$.
Recall that $\nu(U)$ is the smallest integer such that 
$$
 \xi^{\nu(U)} \Bo^{-1} U \Bi
$$
has entries in $\Oe$. We immediately see that 
\begin{equation}
\label{eq:nu-consistent}
\nu(gU) \le \nu(g) + \nu(U)
\end{equation}
because matrix 
$$
 \xi^{\nu(U) + \nu(g)} \Bo^{-1} \, g U \, \Bi = \xi^{\nu(g)} \Bo^{-1} \, g \, \Bo \cdot \xi^{\nu(U)} \Bo^{-1} \, U \, \Bi
$$
has entries in $\Oe$.
Moreover, by the basis change property cost-zero isometries are exactly the ones with $\nu(U) = 0$.
We note that the denominator exponent of the channel representation introduced in \cite{Gosset2014} is also a consistent heuristic for Clifford+T exact synthesis.

\subsection{Computing vertex neighbours}
\label{sec:vertex-neighbors}

Here we discuss a few observations that improve efficiency of computing neighbours 
of a given vertex. 
First, we can take advantage of the basis-change property and 
represent generators and isometries using bases given by input and output basis change matrices.
Second, we can reduce computing neighbours to multiplying integer matrices. 
Say $g$ is a generator written as $\tilde g$ in basis $\Bo$
and $c$ is an a column of $\tilde U$ for some isometry $U$.
We can store the column as integer vector using integer coordinates of the column entries in the integral basis of $\Oe$~(\cref{eq:vec}),
that is 
$$
 c = (c_1,\ldots, c_N) \text{ is stored as } \mathrm{vec}(c) = \mathrm{vec}(c_1) \oplus \ldots \oplus \mathrm{vec}(c_N).
$$
The integer coefficients of the product are $\mathrm{vec}(gc) = Z(g) \mathrm{vec}(c)$ where $Z(g)$ is an integer matrix 
define in~\cref{eq:as-integer-matrix}.

When computing the neighbours one can also take advantage of additional structure of the generators.
For example, suppose that we have generators that are indexed by Pauli operators, such as $\exp(i\pi(I-P)/8) = \frac{1+e^{i\pi/4}}{2}I + \frac{1-e^{i\pi/4}}{2}P $.
We can compute $P U$ for all Pauli operators $P$, by reusing the result of computing $(I \otimes P') U$
to compute $(P'' \otimes P') U$ for $P'' \in \{ X,Y,Z \}$. 
When multiplying $(I \otimes P') U$ by $P'' \otimes I$ using integer matrix approach described above, one can take advantage of the fact that $Z( \Bo^{-1} (P'' \otimes I) \Bo ) $ is sparse. 

\section{Advanced heuristic functions}
\label{sec:advanced-heuristics}

The efficiency of A* search algorithm and similar algorithms depends on the heuristic function.
In this section we show that Smith Normal Form for matrices with entries in a principal ideal domain
can be used to construct a variety of heuristic functions.

Let us briefly explain why Smith Normal Form is a useful tool here.
Recall, that by the definition of the exact synthesis~\cref{prob:exact-synthesis}, 
the optimal cost of isometry $U$ and isometry $U' = C U C'$ must be the same 
when $C$ is from cost-zero group $\mathcal{C}$ and 
$C'$ is from the restricted zero-cost group $\mathcal{C}^{(n')}$.
When $\mathcal{C}$ has basis change property with input and output basis change matrices $\Bi, \Bo$,
isometries $U$ and $U'$ written in $\Bi, \Bo$ differ by left and right multiplication by 
the invertible matrices with entries in $\Oe$~(the ring of integers associated with base-ring $\R$ 
of cost-zero group $\mathcal{C}$).
Smith Normal Form is a canonical form a matrix with entries in $\Oe$ that is invariant 
with respect to left and right multiplication of the matrix by invertible square matrices with entries in $\Oe$.
To ensure that heuristic cost estimate $h(U)$ and $h(U')$ have the same value
it is natural to use Smith Normal Form~(\cref{sec:snf}) to construct heuristic $h$.

The following technical lemma about the properties of Smith Normal Forms is the key tool 
for constructing advanced heuristic functions $h$.

\begin{theorem}[Modularity and invariant factors of unitaries]
\label{thm:unitary-invariant-factors}
\textbf{Consider} the following: 
\begin{itemize}
    \item$\xi$-ring $\R$~(\cref{def:xi-ring}) with the ring of integers $\Oe$ and denominator $\xi$,
    \item $2N \times 2N$ matrix $B$ with entries in $\R$ such that $B^\dagger B$ is $\xi$-modular~(\cref{def:xi-modular}),
    \item $2N \times 2N$ unitary $U$ with entries in $\R$.
\end{itemize}
\textbf{Then} unitary $U$ in basis $B$ can be written as 
\begin{equation}
\label{eq:unitary-valuations}
 B^{-1} U B = A_L\,\mathrm{diag}\left(\xi^{k_1},\ldots,\xi^{k_N},\xi^{-k_N},\ldots,\xi^{-k_1}\right)\,A_R 
\end{equation}
for $\Oe$-unimodular~(\cref{def:xi-modular}) matrices $A_L$, $A_R$ and unique non-increasing sequence of non-negative integers $k_1,\ldots,k_N$.
We use $\nuv_B(U)$ for $(k_1,\ldots,k_N)$ and call it \textbf{the coordinates of }$U$ \textbf{with respect to} $B$.
\end{theorem}
\begin{proof}
The proof proceeds by considering Smith Normal From of $\tilde U = \xi^j B^{-1} U B $,
where $j$ is such that $\tilde U$ has entries in $\Oe$.
We then derive the additional properties on the invariant factors of $\tilde U$ from equation $U^\dagger U = I$ 
and the fact that  $B^\dagger B$ is $\xi$-modular.

Unitary $U$ has rank $N$ and so does $\tilde U$, therefore according to~\cref{thm:smith-normal-form}
$$ 
\tilde U = A_L D A_R, D = \mathrm{diag}(d_1,\ldots, d_{2N})
$$
for some $\Oe$-unimodular matrices $A_R,A_L$.
Next substitute $U$ expressed via $\tilde U$ in $U^\dagger U = I$ and multiply by $B^\dagger$ on the left and $B$ on the right to get:
$$
(\xi \xi^\ast)^j B^\dagger B = \tilde U B^\dagger B \tilde U^\dagger
$$

Using that $B^\dagger B$ is $\xi$-modular above equation simplifies to 
$$
(\xi \xi^\ast)^j A_B = \tilde U A_B \tilde U^\dagger,
$$
for some $\Oe$-unimodular matrix $A_B$.
Now we substitute expression for $\tilde U$ in terms of $D$ and repeatedly use the fact that $A'B' = I$ implies $B'A' = I$
for any full-rank square matrices $A' B'$:
\begin{equation}
\label{eq:dd}
(\xi \xi^\ast)^j I = (A_L^\dagger A_B^{-1} A_L) D (A_R A_B A_R^\dagger) D^\dagger 
\end{equation}
We observe that $\det(D)\det(D)^\ast = (\xi \xi^\ast)^{ 2N j } u$ for some unit $u$ of $\Oe$ because determinants of $\Oe$-unimodular matrices are units.
Taking into account that $\xi \Oe$ is a totally ramified ideal we have $\xi = u' \xi^\ast$ for some unit $u'$ of $\Oe$.
We see that each of invariant factors $d_{l}$ divides a power of $\xi$ and must be a power of $\xi$ up to a unit.

Next, let us establish the additional conditions of the invariant factors of unitaries.
We rewrite \cref{eq:dd} as 
$$
(\xi \xi^\ast)^j (D^\dagger)^{-1} = (A_L^\dagger A_B^{-1} A_L) D (A_R A_B A_R^\dagger) 
$$
The right-hand side is an integer matrix with entries in $\Oe$ and so the left-hand side must be an integer matrix.
This implies that $v_{\mathfrak p}(d_j) \le 2 j$.
Matrices $ D$ and $(\xi \xi^\ast)^j (D^\dagger)^{-1}$ must have the same 
invariant factors up to units and \cref{eq:unitary-valuations} holds.
The uniqueness of $\nuv_B(U)$ follows from the fact that invariant factors of a matrix with entries in $\Oe$ are unique up to multiplication by units of $\Oe$.
\end{proof}

Function $\nuv_B(U)$ can be efficiently computed using the Smith Normal Form algorithms discussed in \cref{sec:snf}.
The following lemma shows an important property of $\nuv_B(U)$ that lets us construct consistent heuristics using it.
\begin{lemma}
Consider $2N\times 2N$ unitaries $U,V$ with entries in $\xi$-ring $\R$~(\cref{def:xi-ring}) and $2N\times 2N$ matrix $B$ such that $B^\dagger B$ is $\xi$-modular~(\cref{def:xi-modular}),
then the coordinates of $U,V$ and $UV$ relative to $B$~(as defined in \cref{thm:unitary-invariant-factors}) are related as:
\begin{equation}
\label{eq:product-majorization}
\nuv_B(UV) \wmaj \nuv_B(U) + \nuv_B(V)
\end{equation}
\end{lemma}
\begin{proof}
The result follows from \cref{thm:unitary-invariant-factors} and \cref{lem:invariant-factors-products}.
Let us use notation $\nuv_B(U) = (k_1,\ldots,k_N)$, $\nuv_B(V) = (k'_1,\ldots,k'_N)$
and $\nuv_B(UV) = (k''_1,\ldots,k''_N)$. 
Next consider integer matrices $\tilde U = \xi^{k_1} B^{-1} U B$, $\tilde V = \xi^{k'_1} B^{-1} V B$
and apply \cref{lem:invariant-factors-products} to $\tilde U, \tilde V$ and $\tilde U \tilde V$.
The lemma implies that 
$$
 \xi^{ \sum_{j=1}^{k} k_1 - k_j } \cdot \xi^{ \sum_{j=1}^{k} k'_1 - k'_j } \text{ divides }  \xi^{ \sum_{j=1}^{k} (k_1 + k'_1) - k''_j } \text{ for } k \le N.
$$
Rewriting above as inequality for powers of $\xi$ leads to conditions $ \sum_{j=1}^{k}  k''_j  \le \sum_{j=1}^{k} (k_j + k'_j) $, 
which shows \cref{eq:product-majorization} by the definition of the weak majorization in~\cref{eq:weak-majorization}.
\end{proof}

Above lemma implies that any positive convex combination of partial sums of the form 
\begin{equation}
\label{eq:general-consistent-heuristic}
 h_{\Vec{p}}(U) = \sum_{k=1}^N p_k \sum_{j=1}^k \nuv_B(U)
\end{equation}
leads to a consistent search heuristic for exact unitary synthesis, when the cost of generators $g$ is $h_{\Vec{p}}(g)$.
This follows from an additional observation that $h_{\Vec{p}}(U) = 0$ implies $\nu(U)= 0$ and so $h_{\Vec{p}}(U) = 0$ implies that $U$ belongs to a source vertex of the search graph.

For a more general exact isometry synthesis a similar approach leads to an admissible heuristic.
This follows from the following theorem 
\begin{theorem}
[Modularity and invariant factors of isometries]
\label{thm:isometry-invariant-factors}
\textbf{Consider} the following: 
\begin{itemize}
    \item $\xi$-ring $\R$~(\cref{def:xi-ring}) with the ring of integers $\Oe$ and denominator $\xi$,
    \item invertible matrices $\Bi, \Bo$ with entries in $\R$ such that $ \Bo^{-1} \, I_{2N\times N'} \Bi$ has entries in $\Oe$ and unit invariant factors~(\cref{thm:smith-normal-form}),
    additionally $\Bi^\dagger \Bi$, $\Bo^\dagger \Bo$ are $\xi$-modular~(\cref{def:xi-modular}), $I_{2N\times N'}$ is defined in \cref{eq:id-isometry},
    \item unitary matrix $U$ with entries in $\R$.
\end{itemize}
\textbf{Then} Isometry $U' = U I_{2N\times N'}$ in basis $\Bo, \Bi$ can be written as 
\begin{equation}
\label{eq:isometry-valuations}
 \Bo^{-1} U' \Bi = A_L\, \left(\begin{array}{c}
    \mathrm{diag}\left(\xi^{-k_{N'}},\ldots,\xi^{-k_1}\right) \\
    \hline
    \textbf{0}_{(2N-N')\times N'}
 \end{array}\right) \,A_R 
\end{equation}
for $\Oe$-unimodular~(\cref{def:xi-modular}) matrices $A_L$, $A_R$ and unique non-decreasing integers $k_1,\ldots,k_{N'}$. 
We call $\nuv_{\Bi,\Bo}(U') = (k_1,\ldots, k_{\min(N,N')})$ \textbf{the coordinates of} $U'$
\textbf{with respect to} $\Bi,\Bo$. Sequences $\nuv_{\Bi,\Bo}(U')$ and $\nuv_{\Bo}(U)$ are related as 
\begin{equation}
\label{eq:isometry-majorization}
\nuv_{\Bi,\Bo}(U')  \wmaj \nuv_{\Bo}(U)_{[\min(N,N')]},
\end{equation}
where $\nuv_{\Bo}(U)_{[\min(N,N')]}$ are the first $\min(N,N')$ coordinates of $\nuv_{\Bo}(U)$.
\end{theorem}
\begin{proof}
Let us denote $k'_j = \nuv_{\Bo}(U)_{j}$.
Using Smith Normal Form~\cref{thm:smith-normal-form} write $\tilde U = \xi^{k'_1} \Bo^{-1} U' \Bi $ as $A_L \left(\frac{D}{\textbf{0}}\right) A_R$
for a square diagonal matrix $D$.

Our first goal is to establish that diagonal of $D$ consists of powers of $\xi$ up to units of $\Oe$.
We expand $(U')^{\dagger} U' = I_{N'\times N'}$ in terms of $\tilde U, \Bi, \Bo$:
$$
 \tilde U^\dagger \Bo^\dagger \Bo \tilde U  = (\xi \xi^\ast)^{k'_1} (\Bi^\dagger \Bi)
$$
Using $\xi$-modularity of $\Bi^\dagger \Bi$ and $\Bo^\dagger \Bo $, we have 
$$
 \Bi^\dagger \Bi = A_{\inn} \xi^{k_\inn}, \Bo^\dagger \Bo = A_{\out} \xi^{k_\out} \text{ for } \Oe-\text{unimodular } A_{\inn}, A_{\out}
$$
Using above and the expression for $\tilde U$ in terms of $A_L, D, A_R$ we have 
$$
 D^\dagger \left( (I_{N'}|\textbf{0}) A_L^\dagger A_\out A_L \left(\frac{I_{N'}}{\textbf{0}}\right) \right) D = (\xi \xi^\ast)^{k'_1} \xi^{k_\inn - k_\out} (A_R^\dagger)^{-1} A_{\inn} A_R^{-1}
$$
We also note that $k_\inn - k_\out$ is non-negative, because $A_r = \Bo^{-1} \, I_{2N\times N'} \Bi$ 
has entries in $\Oe$ and  
$$
 A_r^\dagger \Bo^{\dagger} \Bo A_r = \Bi^\dagger \Bi \implies A_r^\dagger A_\out  A_r =  \xi^{ k_\inn - k_\out} A_\inn.
$$
Above implies that $\det(D)\det(D)^\ast$ divides $\xi^{N' (2k'_1 + k_\inn - k_\out)}$
and so the diagonal of $D$ must consist of powers of $\xi$ up to units of $\Oe$ and \cref{eq:isometry-valuations} follows.

\cref{eq:isometry-majorization} follows from applying \cref{lem:invariant-factors-products} to matrix $ \xi^{k'_1} \tilde U $
and matrix $\tilde A = \Bo^{-1} (\ket{0_{2N-N'}}\otimes I_{N'}) \Bi$ with entries in $\Oe$ and unit invariant factors.
Product $\tilde U \tilde A$ is $\xi^{k'_1} \Bo^{-1} U' \Bi$.
Using \cref{lem:invariant-factors-products} products of the invariant 
factors of $\xi^{k'_1} \Bo^{-1} U' \Bi $ satisfy the following divisibility conditions:
$$
\prod_{j=1}^k \xi^{k'_1 - k'_j} \text{ divides } \prod_{j=1}^k \xi^{k'_1 - k_j }
$$
and therefore \cref{eq:isometry-majorization} holds.
The uniqueness of $\nuv_{\Bi,\Bo}(U')$ follows from the fact that invariant factors of a matrix with entries in $\Oe$ are unique up to multiplication by units of $\Oe$.
\end{proof}

Using above lemma we can construct a variety of admissible heuristic functions using positive convex combinations of partial sums of $\nuv_{\Bi,\Bo}(U')$: 
\begin{equation}
 h_{\Vec{p},N'}(U') = \sum_{k=1}^{N'} p_k \sum_{j=1}^k \nuv_{\Bo,\Bi}(U)_j,~~ h_{\Vec{p},N'}(U) = \sum_{k=1}^{N'} p_k \sum_{j=1}^k \nuv_{\Bo}(U),
\end{equation}
For an isometry $U'$ that is equal to the product $U_1,\ldots,U_m$ applied to $I_{2N\times N'}$ we have inequality 
$$
 h_{\Vec{p},N'}(U') \le \sum_{j=1}^{m} h_{\Vec{p},N'}(U_j)
$$
Therefore $ h_{\Vec{p},N'}(U')$ is a lower-bound on the cost of implementing isometry $U'$, when the cost of a generator $g$ is defined as $h_{\Vec{p},N'}(g)$.
Observation that $ h_{\Vec{p},N'}(U') = 0$ implies $\nu(U') = 0$ and so  $ h_{\Vec{p},N'}(U') = 0$ implies that $U'$ is source vertex of the search graph.
This establishes that $ h_{\Vec{p},N'}$ is an admissible heuristic.

Another application for the heuristic function  $\nuv_B(U)$ defined in \cref{thm:unitary-invariant-factors} is a simple best-first search exact synthesis algorithm:

\begin{algorithm}[H]
\caption{Synthesis via best-first search}
\label{alg:best-first-search}
\begin{algorithmic}[1]
    \State \Input Unitary $U$ with entries in $\xi$-ring $\R$~(\cref{def:xi-ring}), matrix $B$ such that $B^\dagger B$ is $\xi$-modular~(\cref{def:xi-modular}),
    set $G$ of unitaries with entries in $\R$
    \State \Output Sequence of unitaries $U_1,\ldots,U_M$ from $G$ and remainder unitary $V$ integral in basis $B$ such that $U = V U_M \ldots U_1$
    \State Set $S$ to an empty sequence, set $V$ to $U$
    \While{ $\nuv_B(V) \ne 0 $ } \Comment See \cref{thm:unitary-invariant-factors} for definition of $\nuv_B(V)$
        \For{$g \in G$ } 
            \If{ $\nuv_B(V g^\dagger) \lmaj \nuv_B(V)$ }  \Comment See \cref{sec:majorization} for definition of $\lmaj$ \label{line:lmaj-reduction}
            \State Prepend $g$ to $S$, set $ V \leftarrow V g^\dagger$
            \EndIf
        \EndFor
    \EndWhile
    \State \Return {$S$, $V$}   
\end{algorithmic}
\end{algorithm}

\begin{table}[ht]
    \centering
\begin{tabular}{|c|c|c|c|c|}
\hline 
$\xi$ & Basis $B$ & Field $E$ & Gates $U_1,\ldots,U_m$ & Cost-zero group $\mathcal{C}$ \tabularnewline
\hline 
\hline 
$1+i$ & $B_{\mathbb{C}}^{\otimes2}$ & $\mathbb{Q}(i)$ & $\text{CS}$~(\cref{eq:common-complex}) & Clifford \tabularnewline
\hline 
$\sqrt{2}$ & $B_{\mathbb{R}}^{\otimes2}$ & $\mathbb{Q}(\sqrt{2})$ & $\text{CH}$~(\cref{eq:common-real}) & Real Clifford \tabularnewline
\hline 
\hline 
$1+\zeta_{8}$ & $B_{\mathbb{C}}^{\otimes2}$ & $\mathbb{Q}(\zeta_8)$ & $\text{T},\text{T}^{\otimes2},\text{CT}$~(\cref{eq:common-complex}) & Clifford \tabularnewline
\hline 
$2+2\cos(\frac{\pi}{16})$ & $B_{\mathbb{R}}^{\otimes2}$ & $\mathbb{Q}(\cos(\frac{\pi}{8}))$ & $\text{T}_{y},\text{CS}_{y},\text{CT}_{y}$~(\cref{eq:common-real}) & Real Clifford \tabularnewline
\hline 
\end{tabular}
    \caption{Two-qubit gate sets with a simple best-first synthesis~\cref{alg:best-first-search}.
    The algorithm applies to matrices with entries in $\xi$-rings $\R$ with denominator $\xi$ and ring of integers $\Oe$.
    The set of gates $G$ used by the best-first search algorithm is the normalization of $U_1,\ldots,U_m$ with respect to cost-zero group $\mathcal{C}$~(\cref{sec:normalized-gate-sets-algo}).
    See~\cref{thm:best-first-search} for the algorithm runtime.}
    \label{tab:best-first-search}
\end{table}

The main challenge is to show that this algorithm always terminates, that is we can always find $g$ from $G$ such that 
condition in \cref{line:lmaj-reduction} holds. 
In \cref{sec:provable} we will prove the following result:

\begin{theorem}[Best first search]
\label{thm:best-first-search}
The best first search~\cref{alg:best-first-search} terminates in the number of steps linear in $\max \nuv(U)$ 
for a set of unitaries with entries in $\xi$-rings $\R$, matrices $B$ and gate-sets $G$ described in \cref{tab:best-first-search}.
The  $\xi$-rings $\R$ are the rings with denominator $\xi$ and field $E$ given in the table.
Set of unitaries $G$ is obtained from $U_1,\ldots,U_m$ as the normalization with respect to cost-zero group $\mathcal{C}$ for $U_1,\ldots,U_m,\mathcal{C}$ given in the table.
\end{theorem}

Note that using \cref{eq:product-majorization}, we can upper-bound $\max \nuv_B(U)$ by a constant times the number of non-Clifford gates.
This implies, that for Clifford and $\text{T},\text{T}^{\otimes2},\text{CT}$ gate set the best-first search runtime is proportional to the T-count of the input unitary $U$.

Before we can prove above result, we need will to review additional mathematical and computational results related to Hermitian lattices,
which is the topic of \cref{sec:advanced-preliminaries}.
Before diving into the most technical sections of the paper we summarize some of our numerical results. 

\section{Numerical results}
We have implemented A* search with the heuristic function (\cref{eq:heuristic-practice} below) derived from $\nuv_B(U)$ from~\cref{sec:advanced-heuristics} and gate set normalization 
algorithms in Julia. 
The results of applying our algorithms to various exact synthesis problems are summarized in~\cref{tab:unitary-synthesis-results}, \cref{tab:isometry-synthesis}.
We use Julia package \texttt{AStarSearch.jl} for A* search implementation. 
In practice, we use the following heuristic 
\begin{equation}
\label{eq:heuristic-practice}
     h(U) = 10 \sum_{j=1}^{N} 2^{32(j-1)} (\nuv_B(U))_j 
\end{equation}
that captures all of the values of $\nuv_B(U)$ when whey can be represented using less than thirty two bits.
Multiplicative coefficient $10$ speeds up our A* search.
In practice, we observe that we still get circuits with small number of non-Clifford gates even when using a consistent heuristic re-scaled by factor of $10$.

For runtime, we report precomputation time and synthesis time.
Precomputation time consists of time needed to compute a normalized gate set~(\cref{sec:normalized-gate-sets-algo})
and time needed to represent all the normalized generators as integer matrices~(\cref{sec:vertex-neighbors}).
We use a single thread for the precomputation step.
For the synthesis step we compute vertex neighbours in parallel and use 24 threads. 
All the experiments were performed on a computer with two Intel Xeon Gold 6136 CPUs with total of 24 cores and 128GB of memory.

We reproduce well-known circuit identities in \cref{tab:classic-identities}, \cref{tab:reversible-functions}, \cref{tab:states}.
We apply exact synthesis algorithm to unitaries and isometries obtained from the catalytic embeddings~\cite{amy2024exact} 
in \cref{tab:catalytic-embeddings}, \cref{tab:states}.
We report the runtime of synthesis of three qubit computational basis state permutations in~\cref{tab:permutations} using Clifford and CCZ gate set. 
Finally we report T and CS gate counts for $4\times 2$ isometries 
\begin{equation}
\label{eq:v-isometry}
U'_{a,b,c,d} = \frac{1}{\sqrt{2^{k}}}\left(\begin{array}{c}
aI+ibX+icY+idZ\\
\hline I
\end{array}\right),a,b,c,d\in\mathbb{Z},a^{2}+b^{2}+c^{2}+d^{2}=2^{k}-1    
\end{equation}
which can be used to implement all possible unitaries $\text{V}_{2^k-1}$
$$
U_{a,b,c,d} = \left( aI+ibX+icY+idZ \right)/ \sqrt{2^k-1}
$$
with probability $1-1/(2^k)$.
In \cref{tab:v-isometry-via-cs}, \cref{tab:v-isometry-via-t} we show synthesis runtime when using Clifford and CS, Clifford and T gate sets correspondingly
to implement isometries from \cref{eq:v-isometry}.

\begin{table}[pt]
    \centering
    \begin{subtable}{\textwidth}
        \centering
        \caption{Well-known circuit identities.}
        \label{tab:classic-identities}
        \begin{tabular}{|c|c|c|c|c|}
            \hline 
            Common & Number  & Gate set & Precomputation & Synthesis\tabularnewline
            identity  & of qubits & Clifford and $\ldots$ & time (sec.) & time (sec.)\tabularnewline
            \hline 
            \hline 
            CH via 2 T & 2 & T & 0.06 & 0.85\tabularnewline
            \hline 
            CS via 3 T & 2 & T & 0.30 & 0.02\tabularnewline
            \hline 
            CT via 3 $\sqrt{\text{T}}$ & 2 & T, $\sqrt{\text{T}}$ & 1.02 & 0.49\tabularnewline
            \hline 
            CCZ via 3 CS & 3 & CS & 20.68 & 18.02\tabularnewline
            \hline 
            CCZ via 7 T & 3 & T & 8.40 & 40.66\tabularnewline
            \hline 
        \end{tabular}

    \end{subtable}
    \begin{subtable}{\textwidth}
        \centering
        \caption{Unitaries obtained via catalytic embeddings.}
        \label{tab:catalytic-embeddings}
        \begin{tabular}{|c|c|c|c|c|}
        \hline 
        Catalysis & Number  & Gate set & Precomputation & Synthesis\tabularnewline
        example & of qubits & Clifford and $\ldots$ & time (sec.) & time (sec.)\tabularnewline
        \hline 
        \hline 
        S via 2$\text{S}{}_{y}$ & 2 & $\text{CS}{}_{y}$ & 0.02 & 0.003\tabularnewline
        \hline 
        T via CS & 2 & T & 0.06 & 0.008\tabularnewline
        \hline 
        T via $\text{CS}{}_{y}$ & 2 & $\text{CS}{}_{y}$ & 0.03 & 0.006\tabularnewline
        \hline 
        2 T via 2 CCZ & 3 & CS, CCZ & 28.42 & 0.12\tabularnewline
        \hline 
        2 T via $\text{CS}{}_{y}$ & 3 & $\text{CS}{}_{y}$, CCZ & 11.08 & 0.03\tabularnewline
        \hline 
        $\sqrt{\text{T}}\text{H}\sqrt{\text{T}}$ via 5 T & 2 & T, $\text{T}\otimes\text{T}$ & 0.72 & 0.04\tabularnewline
        \hline 
        2 $\sqrt{\text{T}}$ via T, 2 CCZ & 3 & T, $\text{T}\otimes\text{T}$, CS, CCZ & 136.86 & 3.17\tabularnewline
        \hline 
        CS via 2 CCZ & 3 & $\text{CS}{}_{y}$, CCZ & 11.91 & 0.07\tabularnewline
        \hline 
        \end{tabular}
    \end{subtable}
    \begin{subtable}{\textwidth}
        \centering
        \caption{All equivalence classes of $3$-qubit computational basis permutations up to left and right composition with X,CX circuits.
        The permutations are synthesized using Clifford and CCZ gate set.}
        \label{tab:permutations}
        \begin{tabular}{|c|c|c|}
        \hline 
        Permutation & Number of  & Synthesis\tabularnewline
        {[}0,1,2,3,4,5,6,7{]} to & CCZ gates & time (sec.)\tabularnewline
        \hline 
        \hline 
        {[}0, 1, 2, 7, 4, 5, 6, 3{]} & 1 & 0.005\tabularnewline
        \hline 
        {[}0, 1, 2, 5, 4, 7, 6, 3{]} & 2 & 0.056\tabularnewline
        \hline 
        {[}0, 1, 2, 5, 4, 6, 7, 3{]} & 3 & 0.065\tabularnewline
        \hline 
        \end{tabular}
    \end{subtable}
    \begin{subtable}{\textwidth}
        \centering
        \caption{Common $3$-qubit reversible unitaries synthesized using Clifford and CCZ gate set.}
        \label{tab:reversible-functions}
        \begin{tabular}{|c|c|c|}
        \hline 
        Reversible function & Number of & Synthesis\tabularnewline
         & CCZ gates & time (sec.)\tabularnewline
        \hline 
        \hline 
        $|x,y,z\rangle\mapsto|x,y,z + xy\rangle$ & 1 & 1.113\tabularnewline
        \hline 
        $|x,y,z\rangle\mapsto|x,y,z + (x\text{ or }y)\rangle$ & 1 & 0.033\tabularnewline
        \hline 
        $|x,y,z\rangle\mapsto|x,y,z + (x\text{ nand }y)\rangle$ & 1 & 0.026\tabularnewline
        \hline 
        Controlled-SWAP & 1 & 0.030\tabularnewline
        \hline 
        \end{tabular}
    \end{subtable}
    \caption{Runtime of unitary synthesis via A* search with heuristic in~\cref{eq:heuristic-practice}. See \cref{sec:common-matrices} for the notation for common unitaries, states and isometries.}
    \label{tab:unitary-synthesis-results}
\end{table}

\begin{table}[pt]
    \centering
    \begin{subtable}{\textwidth}
        \centering
        \caption{Common states and isometries.}
        \label{tab:states}
        \begin{tabular}{|c|c|c|c|c|c|}
        \hline 
                       & \multicolumn{2}{c|}{Number of} & Gate set & Pre-        &          \tabularnewline
        Isometries and & \multicolumn{2}{c|}{qubits} & Clifford & computation & Synthesis\tabularnewline
        \cline{2-3} \cline{3-3} 
        states  & input & output & and $\ldots$ & time (sec.) & time (sec.)\tabularnewline
        \hline 
        \hline 
        $|\text{CCZ}\rangle$ via 2 CS & 0 & 3 & CS & 19.08 & 1.06\tabularnewline
        \hline 
        $|\text{CS}\rangle$ via 3 T & 0 & 2 & T & 0.07 & 0.36\tabularnewline
        \hline 
        \hline 
        Catalyzed $|\sqrt{\text{T}}\rangle$ via 5T & 1 & 2 & T, $\text{T}\otimes\text{T}$ & 0.68 & 0.03\tabularnewline
        \hline 
        $|x,y\rangle\mapsto|x,y,xy\rangle$ via 4T & 2 & 3 & T & 6.71 & 156.07\tabularnewline
        \hline 
        \end{tabular}\\
 \end{subtable}
    \begin{subtable}{0.45\textwidth}
        \centering
        \caption{$\text{V}_{2^k-1}$ gates via CS gates, \cref{eq:v-isometry}}
         \label{tab:v-isometry-via-cs}
        \begin{tabular}{|c|c|c|}
        \hline 
        V gates  & Non-Clifford & Total\tabularnewline
        family & gates & time (sec.)\tabularnewline
        \hline 
        \hline 
        $\text{V}_{3}$ & 2 CS & 0.11\tabularnewline
        \hline 
        $\text{V}_{7}$ & 4 CS & 0.08\tabularnewline
        \hline 
        $\text{V}_{15}$ & 6 CS & 1.02\tabularnewline
        \hline 
        $\text{V}_{31}$ & 8 CS & 3.05\tabularnewline
        \hline 
        $\text{V}_{63}$ & 10 CS & 36.17\tabularnewline
        \hline 
        \end{tabular}
 \end{subtable}
  \begin{subtable}{0.45\textwidth}
        \centering
        \caption{$\text{V}_{2^k-1}$ gates via T gates, \cref{eq:v-isometry}. }
        \label{tab:v-isometry-via-t}
\begin{tabular}{|c|c|c|}
\hline 
V gates  & Non-Clifford & Total\tabularnewline
family & gates & time (sec.)\tabularnewline
\hline 
\hline 
$\text{V}_{3}$ & 4 T & 0.96\tabularnewline
\hline 
$\text{V}_{7}$ & 6 T & 0.58\tabularnewline
\hline 
$\text{V}_{15}$ & 8 T & 1.67\tabularnewline
\hline 
$\text{V}_{31}$ & 10 T & 2.13\tabularnewline
\hline 
$\text{V}_{63}$ & 12 T & 6.27\tabularnewline
\hline 
\end{tabular}
 \end{subtable}
    \caption{Runtime of isometry synthesis and state preparation algorithms. See \cref{sec:common-matrices} for the notation for common unitaries, states and isometries.}
    \label{tab:isometry-synthesis}
\end{table}

\section{Advanced preliminaries}
\label{sec:advanced-preliminaries}

We assume that the reader is familiar with number fields, their ring of integers, ideals and fractional ideals.
For a review of these topics see Chapters 4.1.2, 4.6.1 in~\cite{Cohen1993}.


\subsection{Hermitian lattices}
\label{sec:hermitian-lattices}

Consider $E$ to be either CM-field or a totally-real field with ring of integers $\Oe$. 
When $E$ is a CM field, we use $^\ast$ for the complex conjugation in $E$. 
When $E$ is totally real, $^\ast$ is the identity on $E$.
We use $E^N$ for $N$-dimensional vector space with inner-product $\langle x, y \rangle = \sum_{j=1}^{N} x_j y_j^\ast$.
With the above inner product, $E^N$ is an example of a totally positive definite Hermitian space.
Let $b_1, \ldots, b_N$ be a basis of $E^N$ and $\a_1,\ldots, \a_N$ be fractional ideals of $\Oe$.
\textbf{Hermitian lattice} in $E^N$ is the following set
$$
 \a_1 b_1 + \ldots + \a_N b_N.
$$
The list $(\a_1,b_1),\ldots,(\a_N,b_N)$ is \textbf{a pseudo-basis} of the lattice.
The set $B \Oe^N$ for invertible $N\times N$ matrix with entries in $E$ is an example of Hermitian lattice with \textbf{basis matrix} $B$.
When $\Oe$ is a principal ideal domain, every lattice has a basis.
For more details on Hermitian lattices see~\cite{Kirschmer2016}.
For any unitary matrix $U$ with entries in $E$, lattice $U L$ is the lattice with pseudo-basis 
$(\a_1,U b_1),\ldots,(\a_N,U b_N)$.
Two Hermitian lattices $L_1, L_2$ are \textbf{isometric} when there there exist a unitary $U$ over $E$ such that $L_1 = U L_2$.
We say that the unitary $U$ \textbf{preserves} $L$ if $UL = L$, such unitaries form a group, \textbf{the automorphism group} of $L$.
There is a practical algorithm~(Remark~2.4.4 in~\cite{Kirschmer2016}, \cite{Plesken1997}) for checking if two lattices are isometric, that is for solving 
\begin{problem}[Hermitian Lattice Isomorphism]
\label{prob:isomorphism}
Given Hermitian lattices $L,M$ in $E^N$, find unitary $U$ with entries in $E^N$ if it exists.
\end{problem}
An implementation of an algorithm for solving above problem is available in \cite{Magma,Nemo}.
There is also a practical algorithm for finding automorphism group~(Remark~2.4.4 in~\cite{Kirschmer2016}, \cite{Plesken1997}) of a Hermitian lattice, 
with implementations in \cite{Magma,Nemo}:

\begin{problem}[Hermitian Lattice Automorphisms]
\label{prob:automorphism}
Given a Hermitian lattice $L$ in $E^N$, find unitaries $U_1,\ldots,U_m$ with entries in $E^N$ that generate 
the automorphism group of $L$.
\end{problem}

For any fractional ideal $\mathfrak{f}$, lattice $\mathfrak{f} L$ is a lattice with pseudo-basis  $(\mathfrak{f} \a_1,b_1),\ldots,(\mathfrak{f} \a_N,b_N)$.
Similarly for any $\alpha$ from $E$, the lattice $\alpha L = (\alpha \Oe) L$ where fractional ideal $\mathfrak{f} = \alpha \Oe$.



\subsection{Local rings and fields}

Consider number field $E$ and its ring of integers $\Oe$. 
For any prime ideal $\p$ of $\Oe$, and any element $x$ of $E$ define 
$|x|_\p = (1/N(\p))^{v_\p(x)}$ where $v_\p(x)$ is the power of $\p$ in the factorization of the fractional ideal $x\Oe$ into prime ideals
and $N(\p)$ is the norm of $\p$. Function $|x|_\p$ is an example of a non-Archimedean norm and the following inequality holds:
$$
 |x+y|_\p  \le \max(|x|_\p,|y|_\p) \le |x|_\p  + |y|_\p
$$
We use notation $E_\p$ for the completion of $E$ with respect to $|\cdot|_\p$ and $\Op$ for the completion of $\Oe$  with respect to $|\cdot|_\p$.
Completion $E_\p$ is a field, $\Op$ is the ring of integers of $E_\p$.
Ring $\Op$ has the unique principal prime ideal that we will denote by $\p$, and any fractional $\Op$ ideal is an integer power of $\p$.
For a review of $\p$-adic fields and related results see Chapter~5 in~\cite{Narkiewicz2004}.

Elements of $\Oe$ and $E$ can also be considered elements of $\Op$ and $E_\p$. 
For a fractional ideal $\a$ of $O_E$ we define its localisation as 
\begin{equation}
\label{eq:ideal-localization}
\a_{\p} = \p ^{v_\p(\a)}, 
\end{equation}
where in the right-hand side $\p$ is the unique prime ideal of $\Op$ and $v_\p(\a)$ is the power of $\Oe$-ideal $\p$
in the factorisation of $\a$ into powers of prime ideals.
We also use convention $\p^{0} = \Op$ in \cref{eq:ideal-localization},
for properties of ideal localization see Chapter~81E in~\cite{OMearaQuad}.
Here we list some of the identities involving $\a_{\p}$
\begin{equation}
\label{eq:ideal-localization-props} 
(\alpha \a)_\p = \alpha \a_\p, ~~ (\a \mathfrak{b})_\p = \a_\p \mathfrak{b}_\p, ~~ (\a + \mathfrak{b})_\p = \a_\p + \mathfrak{b}_\p ~~ (\a \cap \mathfrak{b})_\p = \a_\p \cap \mathfrak{b}_\p, |\a_\p|_\p = |\a|_\p
\end{equation}
\begin{equation}
\label{eq:ideal-localization-props-2} 
\a \subseteq \mathfrak{b} \text{ if and only if } \a_\p \subseteq \mathfrak{b}_\p \text{ for all prime ideals } \p \text{ of } \Oe
\end{equation}

\subsection{Localization of lattices}

Similarly to Hermitian lattices in $E^N$, we can define Hermitian lattices in $E_\p^N$. 
Complex conjugation extends from $E$ to $E_\p$, and the inner product on $E^N$ extends to the inner product on $E_\p$.
Hermitian lattice in $E_\p^N$ is a set 
$$
 \a_1 b_1 + \ldots + \a_N b_N.
$$
for basis $b_1,\ldots,b_N$ of $E_\p^N$ and fractional ideals $\a_j$ of $\Op$.
Given a Hermitian lattice $L$ in $E^N$ with pseudo-basis $(\a_1,b_1),\ldots,(\a_N,b_N)$  
we define its localisation as Hermitian lattice in $E_\p^N$ using \cref{eq:ideal-localization}
$$
 L_\p = (\a_1)_\p b_1 + \ldots + (\a_N)_\p b_N
$$
When lattice $L$ has a basis $B$, the localisation is $L_\p = B \Op^N$, where we consider $B$ as matrix with entries in $E_\p$.

There are few important properties that connect lattices and their localization at prime ideals (See Chapter~81E in~\cite{OMearaQuad}).
For any two Hermitian Lattices the following holds: 
$$
L \subseteq M \text{ if and only if } L_\p \subseteq M_\p \text{ for all prime ideals } \p \text{ of } \Oe
$$
and similarly 
$$
L = M \text{ if and only if } L_\p = M_\p \text{ for all prime ideals } \p \text{ of } \Oe
$$

Moreover, we can modify a Hermitian lattice in $E^N$ at prime ideal $\p$ to construct a new lattice, while leaving all other localizations unchanged:
\begin{theorem}[Special case of Theorem~81:14, Chapter~81E \cite{OMearaQuad}]
\label{thm:local-global-lattice}
\textbf{Consider} a Hermitian lattice $L$ in $E^N$, prime $\Oe$-ideal $\p$,
a Hermitian lattice $\mathfrak{L}$ in $E_\p^N$, \textbf{then} there exist a 
Hermitian lattice $M$ in $E^N$ such that $M_\p = \mathfrak{L}$
and $M_\mathfrak{q} = L_\mathfrak{q}$ for all prime $\Oe$-ideals $\mathfrak{q}$ not equal to $\p$.
\end{theorem}
There is also an efficient algorithm to construct lattice $M$ described in the above theorem 
(see Section 2.2, discussion before Algorithm~2.2.7 in~\cite{Kirschmer2016}).

Above result motivates the following definition 
\begin{definition}
\label{def:p-isometric} 
Let $E$ be a number field.
Hermitian lattice $L,M$ in $E^N$ are $\p$\textbf{-isometric} for prime $\Oe$-ideal $\p$,
when $L_\p$ is isometric to $M_\p$ and $L_\mathfrak{q} = M_\mathfrak{q}$ for all prime 
$\Oe$-ideals $\mathfrak{q} \ne \p.$  
\end{definition}

If two Hermitian lattices are isometric with an isometry $U$ that has entries in $\xi$-ring $R$ such that $\xi \Oe = \p$,
then these two lattices are $\p$-isometric, because $U$ has integer entries in $E_\mathfrak{q}$ for $\p \ne \mathfrak{q}$.
Not every two lattices that are $\p$-isometric are also isometric. 
If two lattices are $\p$-isometric and isometric we say that they belong to the same \textbf{isometry class}.
There are only finitely many classes of non-isometric lattices that are $\p$-isometric.
We will refer to this number of classes as \textbf{genus class number}.
Finally, note that if two lattices are $\p$-isometric and isometric, they must be related by a unitary with entries in the corresponding $\xi$-ring $\R$.

\subsection{Invariant factors of local and global lattices}
\label{sec:invariant-factors}

\begin{theorem}[Theorem~81:11, Chapter~81D \cite{OMearaQuad}]
\label{thm:invariant-factors}
Let $L$ and $M$ be Hermitian lattices in $V = E^N$ or $V = E_\p^N$, then there exist a basis 
$b_1,\ldots,b_n$ of $V$, fractional ideals $\{ \a_j \}_{j=1}^N$, $\{ \mathfrak{r}_j \}_{j=1}^N$ (of $\Oe$ or $\Op$ respectively)
such that 
\begin{equation}
\label{eq:invariant-factors}
L =  \a_1 b_1 + \ldots + \a_N b_N, ~ M = \a_1 \mathfrak{r}_1 b_1 + \ldots + \a_N \mathfrak{r}_N b_N   
\end{equation}
and $\mathfrak{r}_N \subseteq \ldots \subseteq \mathfrak{r}_1$. Fractional ideals $\{ \mathfrak{r}_j \}_{j=1}^N$ are unique.
\end{theorem}
Fractional ideals $\{ \mathfrak{r}_j \}_{j=1}^N$ are the \textbf{invariant factors} of $M$ in $L$.
When $M \subseteq L$, fractional ideals $\{ \mathfrak{r}_j \}_{j=1}^N$ are integral ideals, 
and product $[M:L] = \prod_{j=1}^N \mathfrak{r}_j$ is called an \textbf{index ideal} of $M$ in $L$.
When index ideal is $\Oe$, we have $M = L$. 
\begin{proposition}
\label{prop:index-ideals}
For Hermitian lattices $M \subseteq M' \subseteq L$, $[M:L] = [M:M'][M':L]$.
\end{proposition}
Above proposition follows from relation between invariant factors of lattices and their localizations.
Using the fact that $(\a x)_\p = \a_\p x$~(see \cref{eq:ideal-localization-props}) for any $\Oe$-fractional ideal $\a$ and any vector $x$ from $E^N$ 
we have 
\begin{equation}
\label{eq:pseudo-basis-localization}
    (\a_1 b_1 + \ldots + \a_N b_N)_\p = (\a_1)_\p b_1 + \ldots + (\a_N)_\p b_N.
\end{equation}
and therefore the following proposition holds 
\begin{proposition}[Localization of invariant factors]
\label{prop:invariant-factor-loclization}
Consider Hermitian lattices $L$ and $M$ be in $V = E^N$ and let $\{ \mathfrak{r}_j \}_{j=1}^N$ be invariant factors of $M$ in $L$, 
then for any prime $\Oe$-ideal $\p$ ideal localizations~(\cref{eq:ideal-localization}) $\{ (\mathfrak{r}_j)_\p \}_{j=1}^N$ are the invariant factors of $M_\p$ in $L_\p$. 
\end{proposition}
\begin{proof}
The result follows from \cref{eq:pseudo-basis-localization} and \cref{eq:ideal-localization-props-2}.
\end{proof}

The invariant factors~\cref{thm:invariant-factors} is a generalization of Smith Normal Form~(\cref{thm:smith-normal-form}).
\begin{proposition}
\label{prop:invariant-factors-and-snf}
Let $M,L$ be lattice in $E^N$ such that $\Oe$ is a principal ideal domain and $M \subset L$.
Let $A$ be a matrix with entries in $\Oe$ such that $M = A L$, then for  
invariant factors $\{ d_j \}_{j=1}^N$ of $A$~(\cref{thm:smith-normal-form})
and invariant factors $\{ \mathfrak{r}_j \}_{j=1}^N$ of $M$ in $L$
we have $d_j O_E = \mathfrak{r}_j$.
\end{proposition}
\begin{proof}
The proof follow from explicitly constructing basis $b_j$ in \cref{eq:invariant-factors} using the Smith Normal Form of $A$
and then using the uniqueness of the invariant factors.
\end{proof}

Using the above result we provide the proof of~\cref{prop:index-ideals}.
\begin{proof}[Proof of \cref{prop:index-ideals}]
We show that $[M:L]_\p = [M:M']_\p [M':L]_\p$ for all prime $\Oe$-ideals $\p$.
Indeed, completion of $\Oe$ at $\p$ is a principal ideal domain, so lattices $L_\p, M_\p, M'_\p$ have 
bases $B_L, B_M, B_{M'}$ and there exist matrices $A_{M,M'}$, $A_{M',L}$ with entries in $\Op$
such that $B_M = A_{M,M'} B_{M'}$ and $B_{M'} = A_{M',L} B_L$.
Using Smith Normal Form for principal ideal domains~(\cref{thm:smith-normal-form}) we have $[M:L]_\p = \mathrm{det}(A_{M,M'} A_{M',L}) \Op$,
$[M:M']_\p = \mathrm{det}(A_{M,M'}) \Op$, $[M':L]_\p = \mathrm{det}(A_{M,M'}) \Op$,
which shows the required result.
\end{proof}






\subsection{Maximal sublattice and minimal superlattices}

Consider two Hermitian lattices $M, L$ in $E^N$. 
Lattice $M$ is a \textbf{sublattice} of $L$ when $M \subseteq L$, and $L$ is a \textbf{super-lattice} of $M$.
We say that sublattice of $M$ of $L$ is \textbf{a maximal sublattice} of $L$ if for any Hermitian lattice $L'$ in $E^N$ such that 
$M \subseteq L' \subseteq L$ either $M = L'$ or $L' = L$.
Similarly, $L$ is \textbf{a minimal superlattice} of $M$ if for any Hermitian lattice $L'$ in $E^N$ such that 
$M \subseteq L' \subseteq L$ either $M = L'$ or $L' = L$.

There is an efficient Algorithm~2.2.7 in~\cite{Kirschmer2016} to solve the following 
\begin{problem}[Maximal sublattices at prime ideal]
\label{prob:maximal-sublattice}
\textbf{Given} lattice $L$ in $E^N$ and prime $\Oe$-ideal $\p$, \textbf{find} all maximal sub-lattices $M$ of $L$ 
such that $\p L \subseteq M \subset L$.
\end{problem}
There are finitely many solutions to above problem (see correctness proof of Algorithm~2.2.7 in~\cite{Kirschmer2016}).
Similarly, there is an efficient algorithm for the following: 
\begin{problem}[Minimal superlattices at prime ideal]
\label{prob:minimal-superlattice}
\textbf{Given} lattice $L$ in $E^N$ and prime $\Oe$-ideal $\p$, \textbf{find} all minimal superlattices $M$ of $L$ 
such that $L \subseteq M \subset \p^{-1} L$.
\end{problem}
An implementation of algorithm for \cref{prob:maximal-sublattice}, \ref{prob:minimal-superlattice} 
is available in a Magma~\cite{Magma} package for computing with quadratic and hermitian lattices over number fields~\cite{Kirschmer2016code}.
Moreover, the package implements a more efficient version of the above algorithm that finds representatives of 
lattices up-to the automorphism group of $L$.

There is a simple sufficient condition for a lattice to be maximal sublattice.
\begin{proposition}
\label{prop:maximal-and-prime-index}
Consider Hermitian lattices $M,L$ in $E^N$ with $M \subset L$ and index ideal $[M:L] = \p$ for some prime $\Oe$-ideal $\p$,
then $M$ is a maximal sublattice of $L$ and $L$ is a minimal superlattice of $M$.
\end{proposition}
\begin{proof}
Consider lattice $M'$ such that $M \subset M' \subset L$, then the index ideals are related as 
$[M:L] = [M:M'][M':L]$~(\cref{prop:index-ideals}) and either $[M':L] = \Oe$ or $[M,M'] = \Oe$, which implies that
either $M'=L$ or $M=M'$.
\end{proof}

\subsection{Norm, scale, dual of a lattice and modular lattices}
\label{sec:norm-scale-modular}

\begin{definition}[Scale and norm of a lattice \cite{Kirschmer2016}]
\label{def:scale-and-norm}
Consider field $E$ that is a number field with complex conjugation or a completion of such a number filed at a prime ideal.
Let $K$ be the sub-field of $E$ fixed by the complex conjugation.
Consider lattice $L$ in $E^N$, the \textbf{lattice norm} $\n(L)$ is an fractional $\mathcal{O}_{K}$-ideal  generated by $ \{ \langle x,x \rangle : x \in L \}$,
the \textbf{lattice scale} $\s(L)$ is a fractional  $\Oe$-ideal $\langle L, L \rangle =  \{ \langle x,y \rangle : x,y \in L \} $.
\end{definition}

In terms of pseudo-basis $\{(\a_j,b_j)\}_{j=1}^N$ of the lattice $L$, the scale and norm can be expressed as~(Remark 2.3.4~\cite{Kirschmer2016}):
\begin{equation}
\label{eq:norm-and-scale}
\s(L) = \sum_{1 \le j,k \le m} \a_j \langle b_j, b_k \rangle \a_k^\ast,~~\n(L) = \sum_{1 \le j \le m} \mathrm N(\a_j) \langle b_j, b_j \rangle + \sum_{ 1 \le j < k \le m } \mathrm T(\a_j \langle b_j, b_k \rangle \a_k^\ast ),
\end{equation}
where $\mathrm N$ and $\mathrm T$ map fractional-$\Oe$ ideals to fractional $\mathcal{O}_{K}$-ideals, 
as follows
$$
\mathrm N(\a) \text{ is ideal generated by } \{ xx^\ast : x \in \a \},~~\mathrm T(\a) \text{ is ideal generated by } \{ x + x^\ast : x \in \a \}
$$
For lattice $L$ in $E^N$ Lattice $L^\# = \{ x \in E^N : \langle x, L \rangle \subseteq \Oe \}$ is called the dual of $L$.
Given a pseudo-basis $\{ \a_j, b_j \}_{j=1}^N$ of lattice $L$ the pseudo-basis of its dual is~(Remark 2.3.4~\cite{Kirschmer2016}):
\begin{equation}
\label{eq:dual-pseudo-basis}
     \{ (\a_j^\ast)^{-1}, \bar b_j \}_{j=1}^N \text{ where } \{ \bar b_j \}_{j=1}^N \text{ is a \textbf{dual basis} of in } E^N \text{, that is } \langle b_i, \bar b_j \rangle = \delta_{i,j}  
\end{equation}
For a fractional $\Oe$-ideal $\mathfrak{F}$, we say that lattice is $\mathfrak{F}$-modular when $\mathfrak{F} L^\# = L$.

\subsection{Hyperbolic lattices over local fields}

\subsubsection{Hyperbolic spaces over local fields}

\begin{definition}[Hyperbolic space]
\label{def:hyperbolic-space}
Let $E$ be a CM field or a totally real field and let $\p$ be a prime $\Oe$-ideal. 
We say that vector space $E_\p^{2N}$ is a \textbf{hyperbolic space} if it has a basis $\{ e_j, f_j \}_{j=1}^{N}$, a \textbf{hyperbolic basis}, such that 
$\langle e_j , e_k \rangle = 0$, $\langle f_j , f_k \rangle = 0$ and  $\langle e_j , f_k \rangle = \delta_{j,k}$.
\end{definition}

The vector $e$  such that $\langle e,e \rangle = 0$ is called an \textbf{isotropic vector}.
For a vector space to be hyperbolic it is necessary that it has an isotropic vector. 

\begin{lemma}[Dyadic hyperbolic spaces]
\label{lem:dyadic-hyperbolic-spaces}
Let $E= \q(\zeta_{2^m})$, let $\p$ be the unique prime $\Oe$-ideal with norm $2$, 
let $K$ be the the subfield of $E$ fixed by the complex conjugation and let $\p_K$ be the unique prime $\mathcal{O}_K$-ideal with norm $2$. 
The following vector spaces are hyperbolic: 
\begin{itemize}
    \item $E_\p^{2^n} = \q(\zeta_{2^m})_\p^{2^n}$ when $n \ge 1$ and $m = 3,4,5$ 
    \item $E_\p^{2^n} = \q(i)^{2^n}_\p$ when $n \ge 2$ and $m = 2$
    \item $K_{\p_K}^{2^n} = (\q(\zeta_{2^m}) \cap \mathbb{R})_{\p_K}^{2^n} $ when $n \ge 2$ and $m = 3,4,5$ 
    \item $K_{\p_K}^{2^n} = \q_2^{2^n}$ when $n \ge 3$ and $m = 2$ 
\end{itemize}
Furthermore, $\q(i)^2_\p$, $\q(\sqrt{2})^2_{\p_K}$, $\q_2^2$, $\q_2^4 $ have no isotropic vectors.
\end{lemma}

Above is a well-known result. Here we provide an elementary proof by using local Hilbert symbol. 
\begin{definition}[Local Hilbert symbol]
Let $K$ be a totally real field and let $\p$ be a prime $\mathcal{O}_K$-ideal, \textbf{local
Hilbert symbol} $(a,b)_\p$ is $1$ if $a x^2 + b y^2 = 1$ has a solution $(x,y)$ in $K_\p^2$, 
and $-1$ otherwise.
\end{definition}
There exist an efficient algorithm to compute the local Hilbert symbol~\cite{Voight2013}
and its implementation is available in Hecke~\cite{Nemo} and Magma~\cite{Magma} computer algebra systems.
For example we find the following values of the Hilbert symbol via direct computation: 
\begin{itemize}
    \item $(-1,-1)_\p = 1$ when $\p$ is the prime ideal with norm $2$ in $\q(\sqrt{2})$, $\q(\zeta_{2^m}) \cap \mathbb{R}$, $m=4,5$ 
    \item $(-1/3,-1/3)_\p = 1$ when $\p = 2 \mathbb{Z}_2$ in $\q_2$
\end{itemize}

Using above values of the local Hilbert symbol we can explicitly construct a hyperbolic basis 
for vector spaces mentioned in \cref{lem:dyadic-hyperbolic-spaces}
\begin{proof}[Proof~of~\cref{lem:dyadic-hyperbolic-spaces}]
We first construct basis for $E_\p^{2^n}$ $n \ge 1$ and $m = 3,4,5$. 
Consider the case $n=1$.
There exist $a,b$ in $K_{\p_K}$ such that $a^2 + b^2= -1$,
and so there is $\alpha = a + i b$ in $E_\p$ such that $\alpha \alpha^\ast = -1$.
Let $e_1 = (1, \alpha)$ and $f_1 = (1, -\alpha)/2$. 
For $n \ge 2$, we construct hyperbolic basis using the standard orthogonal basis $\{ b_j \}_{j=1}^{2^{n-1}}$ of $E_\p^{2^{n-1}}$ as $ \{ e_1 \otimes b_j, f_1 \otimes b_j \}_{j=1}^{2^{n-1}}$.

Second, we construct basis for $\q(i)_\p^{2^n}$ $n \ge 2$ and $m = 2$. 
There exist $c,d$ in $\q_2$ such that $c^2 + d^2 = -3$, and so there is $\beta = c + id $ in $\q(i)_\p$ such that $\beta \beta^\ast = -3$.
Using $\beta$  we construct hyperbolic basis in $\q(i)_\p^4$ as  
$$
e_1 = (\beta, 1+i, 1, 0),~f_1 = (-\beta, 1+i, 1, 0)/6,~e_2 = (0, -1, 1-i, \beta),~f_2 = (0, -1, 1-i, -\beta)/6.
$$
For $n \ge 3$, we construct hyperbolic basis using the standard orthogonal basis $\{ b_j \}_{j=1}^{2^{n-2}}$ of $E_\p^{2^{n-2}}$ as $ \{ e_1 \otimes b_j, f_1 \otimes b_j, e_2 \otimes b_j, f_2 \otimes b_j \}_{j=1}^{2^{n-2}}$.

Third, we construct basis for $(\q(\zeta_{2^m}) \cap \mathbb{R})_{\p_K}^{2^n} $ when $n \ge 2$ and $m = 3,4,5$. 
There exist $a,b$ in $K_{\p_K}$ such that $a^2 + b^2= -1$, and using them we construct hyperbolic basis of 
$(\q(\zeta_{2^m}) \cap \mathbb{R})_{\p_K}^{4}$: 
$$
e_1 = (a, b, 1, 0),~f_1 = (-a, -b, 1, 0)/2,~e_2 = (b, -a, 0, 1),~f_2 = (-b, a, 0, 1)/2.
$$
For $n \ge 3$, we construct hyperbolic basis using tensor product.

It remains to construct a hyperbolic basis for  $\q_2^{2^n}$ when $n \ge 3$.
Using that there exist $c,d$ in $\q_2$ such that $c^2 + d^2 = -3$, for $\q_2^{8}$ we have the following basis:
$$
\begin{array}{ll}
e_1 = (c,d,1,1,1,0,0,0)  & f_1 = (-c,-d,1,1,1,0,0,0)/6 \\
e_2 = (0,0,0,-1,1,1,c,d) & f_2 = (0,0,0,-1,1,1,-c,-d)/6 \\
e_3 = (-d,c,-1,1,0,1,0,0)  & f_3 = (d,-c,-1,1,0,1,0,0)/6 \\
e_4 = (0,0,1,0,-1,1,-d,c) & f_4 = (0,0,1,0,-1,1,d,-c)/6 \\
\end{array}
$$
For $n \ge 4$, we construct hyperbolic basis using tensor product.

Finally, the fact that there are no isotropic vectors can be checked using Hecke~\cite{Nemo}
function \texttt{is\_isotropic}, Magma~\cite{Magma} function \texttt{IsIsotropic}, or by checking the conditions described in 
Theorem~3.2.2 and Corollary~3.2.4 in~\cite{Kirschmer2016}.
\end{proof}

\subsubsection{Modular lattices over dyadic completions of CM fields}

The goal of this section it to describe the localisation of Barnes-Wall lattices in $E^{2^n} = \q(\zeta_{2^m})^{2^n}$ at the prime $\Oe$-ideal $\p$
with norm $2$. 
We construct examples of lattices in $E_\p^{2N}$ and then check that the examples we constructed are isomorphic  
to the localizations of Barnes-Wall lattices using the following result:

\begin{theorem}[Corollary~3.3.19~\cite{Kirschmer2016}]
\label{thm:hermitian-isometry}
Consider CM-field $E$ and prime $\Oe$-ideal $\p$ such that $\p$ is ramified relatively to the totally real subfield $K$ of $E$.
Two $\p^i$-modular~(\cref{sec:norm-scale-modular}) lattices $L,M$ in $E_\p^N$ are isometric~(\cref{sec:hermitian-lattices}) if and only if $\n(L) = \n(M)$~(\cref{def:scale-and-norm}).
\end{theorem}

Note that if a lattice is $\p^i$-modular then its scale is $\p^i$~(Remark~2.3.4 in~\cite{Kirschmer2016}).
In other words we can say that two modular lattices are isometric if they have the same norm and scale.
The relevant examples of the lattices are:

\begin{definition}[Scaled hyperbolic lattice]
\label{def:hyperbolic-lattice}
Let number field $E$ be either CM-field or a totally-real field, and let $\p$ be a prime $\Oe$-ideal.
A Hermitian lattice in $E_\p^{2N}$~(with the standard Hermitian inner product) is a \textbf{scaled hyperbolic lattice} when it has a basis $\{ e_j, \alpha f_j \}_{j=1}^N$ for some hyperbolic basis $\{ e_j, f_j \}_{j=1}$~(\cref{def:hyperbolic-space}) with the \textbf{scale} $\alpha = 1$ or $\alpha$ from $\Op$ such that $\alpha \Op = \p$ for $\Op$ the completion of $\Oe$ at $\p$. When $\alpha = 1$ we call the lattice hyperbolic.
\end{definition}

Note that $E_\p^{2N}$ has a hyperbolic lattice if and only if it is a hyperbolic space~\cref{def:hyperbolic-space}. 
The hyperbolic spaces relevant to localizations of Barnes-Wall lattices are described in~\cref{lem:dyadic-hyperbolic-spaces}.
We can check that hyperbolic and scaled hyperbolic lattices are modular lattices 
using \cref{eq:dual-pseudo-basis} and the fact that hyperbolic basis $\{ f_j,e_j \}_{j=1}^N$ is dual to $\{ e_j,f_j \}_{j=1}^N$.

\begin{lemma}[Norm and scale of hyperbolic lattices]
\label{lem:complex-hypernbolic-norm-and-scale}
Let  $E = \q(\zeta_{2^m})$, let $K = E \cap \mathbb{R}$, let $\p$ and $\p_K$ be the prime  $\Oe$-ideal and the prime $\mathcal{O}_K$ ideal with norm $2$.
Suppose that $m,n$ are chosen such that $E^{2^n}_\p$ is a hyperbolic space~(\cref{lem:dyadic-hyperbolic-spaces}).
Let $L$ be a scaled hyperbolic lattices in $E^{2^n}_\p$ with scale $\alpha$, 
then $\s(L) = \alpha \Op$, $\n(L) = \p_K $.
\end{lemma}
\begin{proof}
Using \cref{eq:norm-and-scale},  the fact that for any fractional ideal $I$ the sum $I+I = I$,
 $\pi \Op = \pi^\ast \Op = \p^\ast = \p$ we have  $\s(L) = \alpha \Op$.
Again using~\cref{eq:norm-and-scale}, we have $\n(L) = \mathrm{T}(\alpha \O_p)$.
We use  $\mathrm{T}(\O_p) = \mathrm{T}(\Oe)_{\p_K}$ and $\mathrm{T}(\p) = \mathrm{T}(\p)_{\p_K}$ where in the right-hand side 
we consider $\p$ as $\Oe$-ideal. 
By direct computations in Hecke~\cite{Nemo} we have $\n(L)= \p_K$.
Alternatively, see the calculation for the trace ideal via inverse different of $E/K$ in the beginning of Section~3.3.3 in~\cite{Kirschmer2016}.
\end{proof}

We have used Hecke~\cite{Nemo} to compute scale and norm of the re-scaled Barnes-Wall lattices over $\q(\zeta_{2^m})$.
The result of our computation is the following: 

\begin{lemma}[Localization of Barnes-Wall lattices]
\label{lem:bw-lattice-localization}
Let  $E = \q(\zeta_{2^m})$, let $K = E \cap \mathbb{R}$, let $\p$ and $\p_K$ be the prime  $\Oe$-ideal and the prime $\mathcal{O}_K$ ideal with norm $2$.
Consider Hermitian lattice $\mathcal{L}_n = (1+\zeta_8)^n B_{\mathbb{C}}^{\otimes n} \mathbb{Z}[\zeta_{2^m}]^{2^n}$~(for $B_{\mathbb{C}}$ see \cref{eq:basis-changes}) for $m=3,4,5$, then
scale $\s((\mathcal{L}_n)_\p) = \Oe$, norm $\n((\mathcal{L}_n)_\p) = \p_K$ for $n=1,\ldots,4$ and therefore $(\mathcal{L}_n)_\p$ is isometric to a hyperbolic lattice~(\cref{def:hyperbolic-lattice}).

Consider Hermitian lattice $\mathcal{M}_n = (1+i)^{\lceil n/2 \rceil } B_{\mathbb{C}}^{\otimes n} \mathbb{Z}[i]^{2^n}$
and let $\p' = (1+i) \mathbb{Z}[i]$
then $\s((\mathcal{M}_n)_{\p'}) = (\p')^{n\,\mathrm{mod}\,2}$, $\n((\mathcal{M}_n)_{\p'}) = 2 \mathbb{Z}$ for $n=2,\ldots,5$.
Lattice $(\mathcal{M}_n)_{\p'}$ is isomorphic to a scaled hyperbolic lattice with scale $\alpha = (1+i)^{n\, \mathrm{mod} \, 2}$.
\end{lemma} 
\begin{proof}
The result follow from the direct computation of norm and scale of $\mathcal{L}_n$ and $\mathcal{M}_n$, the fact that $\mathcal{L}_n$ and $\mathcal{M}_n$ are modular lattices,
expressions for norm and scale of scaled hyperbolic lattices~(\cref{lem:complex-hypernbolic-norm-and-scale})
and \cref{thm:hermitian-isometry}.
\end{proof}

One can extend above lemma to other values of $n$ by computing lattice norm and scale analytically using \cref{eq:norm-and-scale}.
One can also extend the lemma to other values of $m$ by analytically computing trace ideal $\mathrm{T}(\Oe)$ for $E = \q(\zeta_{2^m})$.
For our purposes small set of values $m$ and $n$ is sufficient.

\subsubsection{Modular lattices over dyadic completions of totally real fields}
We start by comparing scales and norms of localizations of real Barnes-Wall lattices with scales and norm of hyperbolic lattices in the completions of the corresponding totally real fields.
The equality of scales and norms is generally not sufficient to establish that lattices in the completions of the totally real fields are isometric, 
however the equality is sufficient in the following special case: 
\begin{theorem}
\label{thm:quadratic-isometry}
Consider totally real field $K$ and prime $\mathcal{O}_K$-ideal $\p_K$ such that $\p_K$ has norm $2$.
Consider two modular~(\cref{sec:norm-scale-modular}) lattices $L,M$ in $K_{\p_K}$.
If $\n(L) = \n(M)$, $\s(L) = \s(M)$, $\n(L) = 2\s(L)$~(see \cref{def:scale-and-norm}) then lattices $L$ and $M$ are isometric~(\cref{sec:hermitian-lattices}).
\end{theorem}
\begin{proof}
The proof relies on norm group $\mathfrak{g}(L)$ and weight $\mathfrak{w}(L)$ of a quadratic lattice $L$, see Definition~3.3.7 in \cite{Kirschmer2016}. 
When $\n(L) = 2\s(L)$, Proposition 3.3.8 in \cite{Kirschmer2016} implies that $\mathfrak{w}(L) = \mathfrak{g}(L) = \n(L)$. 
Then, according to Theorem 3.3.10 in \cite{Kirschmer2016}, lattices $L$ and $M$ are isometric because $ \mathfrak{g}(L) =  \mathfrak{g}(M)$.
\end{proof}

The following follows directly from \cref{eq:norm-and-scale}:
\begin{lemma}[Norm and scale of hyperbolic lattices]
\label{lem:real-hyperbolic-norm-and-scale}
Let  $K = \q$ or $K = \q(\zeta_{2^m} + \zeta_{2^m}^{-1})$ for $m \ge 3$, $\p_K$ be the prime $\mathcal{O}_K$-ideal with norm $2$.
Suppose that $m,n$ are chosen such that $K^{2^n}_{\p_K}$ is a hyperbolic space~(\cref{lem:dyadic-hyperbolic-spaces}).
Let $L$ be scaled hyperbolic lattice with scale $\alpha$, 
then $\s(L) = \alpha \mathcal{O}_{\p_K}$, $\n(L) = 2 \alpha \mathcal{O}_{\p_K}$ where $\mathcal{O}_{\p_K}$ is the completion of $\mathcal{O}_K$  at $\p_K$.
\end{lemma}

By direct computation in Hecke~\cite{Nemo} we establish the following:
\begin{lemma}[Norm and scale of localization of real Barnes-Wall lattices]
\label{lem:real-bw-lattice-localization}
Let $K = \q(\zeta_{2^m} + \zeta_{2^m}^{-1})$ for $m = 4,5,6$, $\p_K$ be the prime $\mathcal{O}_K$-ideal with norm $2$.
Consider lattice $\mathcal{L}_n = \left(\sqrt{2 + \sqrt{2}}\right)^n  B_{\mathbb{R}}^{\otimes n} \mathcal{O}_K^{2^n}$~(for $B_{\mathbb{R}}$ see \cref{eq:basis-changes}) for $m=4,5,6$, for $n=2,\ldots,5$, then
scale $\s((\mathcal{L}_n)_{\p_K}) = \mathcal{O}_{\p_K}$, norm $\n((\mathcal{L}_n)_{\p_K}) = 2 \mathcal{O}_{\p_K}$ 
and therefore $(\mathcal{L}_n)_{\p_K}$ is isometric to a hyperbolic lattice~(\cref{def:hyperbolic-lattice}).

Consider lattice $\mathcal{M}_n = \sqrt{2}^{\lceil n/2 \rceil } B_{\mathbb{R}}^{\otimes n} \mathbb{Z}[\sqrt{2}]^{2^n}$ and let $\p' = \sqrt{2} \mathbb{Z}[\sqrt{2}]$ ,
then $\s((\mathcal{M}_n)_{\p'}) = (\p')^{n\,\mathrm{mod}\,2}$, $\n((\mathcal{M}_n)_{\p'}) = 2 (\p')^{n\,\mathrm{mod}\,2}$ for $n=2,\ldots,5$.
Lattice $(\mathcal{M}_n)_{\p'}$ is isometric to a scaled hyperbolic lattice with scale $\alpha = (\sqrt{2})^{n\,\mathrm{mod}\,2}$.
\end{lemma} 
\begin{proof}
The result follow from the direct computation of norm and scale of  $\mathcal{L}_n$ and $\mathcal{M}_n$,
the fact that $\mathcal{L}_n$ and $\mathcal{M}_n$ are modular lattices,
the expression for norm and scale of (scaled) hyperbolic lattices~(\cref{lem:real-hyperbolic-norm-and-scale}) and \cref{thm:quadratic-isometry}.
\end{proof}

Finally, there is also a family of lattices in $\q^{2^n}$ that are isometric to a scaled hyperbolic lattice at $\p  = 2\mathbb{Z}$.
\begin{lemma}[Norm and scale of localization of rational Barnes-Wall lattices]
\label{lem:rational-bw-lattice-localization}
Consider a lattice $L_{n}$ in $\q^{2^{n}}$ with basis $\frac{1}{2^{\lceil n/4 \rceil - 1} } B_{\q,n}$~(\cref{eq:rational-basis}).
For $n=3,5,7,9$, lattice $L_n$ is modular with $\s(L_n) = 2^{((n-1)/2)\,\mathrm{mod}\,2} \mathbb{Z}$
and norm $\n = 2\s$.
The lattice completion at $\p  = 2\mathbb{Z}$ isometric to a scaled hyperbolic lattice.
\end{lemma}
\begin{proof}
The result follows from a direct calculation.    
\end{proof}

Similarly to \cref{lem:bw-lattice-localization}, results \cref{lem:real-bw-lattice-localization,lem:rational-bw-lattice-localization}
can be extended to other values of $m,n$ using analytical tools.

\section{Best-first search efficiency}
\label{sec:provable}

The goal of this section is to prove~\cref{thm:best-first-search} and illustrate the use of computational and mathematical tools introduced in \cref{sec:advanced-preliminaries}.
We envision that these tools will find further applications in circuit synthesis and optimization which we discuss in concluding remarks~\cref{sec:conclusion}.

The crucial property that we need to establish for the proof of~\cref{thm:best-first-search} is the following.
For the basis matrix $B$, the $\xi$-ring $\R$, the set of unitaries $G$ specified in~\cref{thm:best-first-search} (and in \cref{tab:best-first-search}) 
the following is true:
\begin{equation}
\label{eq:lmaj-reduction}
\text{For any } U \text{ such that } \nuv_B(U) \ne 0 \text{ there exists }g \in G \text{ such that } \nuv_B(U g^\dagger) \lmaj \nuv_B(U).
\end{equation}
This ensures that \cref{alg:best-first-search} terminates and has required complexity.

We proceed in two high-level steps.
We first rephrase \cref{eq:lmaj-reduction} in terms of Hermitian lattices
and then show that proving that \cref{eq:lmaj-reduction} is true reduces to checking another related property for a finite set of Hermitian lattices. 

To proceed we require some additional definitions.
For $\xi$-ring $\R$ with ring of integers $\Oe$ and $2N \times 2N$ matrix $B$ 
consider Hermitian lattice $L = B \Oe^{2N}$ in $E^{2N}$. 
Additionally, for unitaries $U_1,U_2$ with entries in $\R$, consider Hermitian lattices $L_1 = U_1 L$ and $L_2 = U_2 L$. 
For any such lattices $L_1,L_2$ define map 
$$
\bv(L_1,L_2) = \nuv_B(U_2^\dagger U_1)
$$
By constructions, lattices $L_1,L_2$ are isometric and $\p$-isometric~(\cref{def:p-isometric}) to $L$. 
We will show that the map $\bv$ can be extended to all lattices $\p$-isometric to $L$ for a range of lattices $L$ relevant to quantum circuit synthesis.
We start with formal definition of $\bv$.

\begin{definition}[Coordinate map]
\label{def:coordinate-map}
Let $L$ be a Hermitian lattice in $E^{2N}$, we say that map $\bv$ from a set of pairs of lattices $\p$-isometric~(\cref{def:p-isometric}) to $L$ to sequences of non-increasing non-negative integers
is a \textbf{coordinate map} associated with $L$ if it has the following properties:
\begin{itemize}
    \item for any $L_1,L_2$ $\p$-isometric to $L$  $\bv(L_1, L_2) = 0$ implies $L_1 = L_2$
    \item for any $L_1,L_2$ $\p$-isometric to $L$  $\bv(L_1, L_2) = \bv(L_2, L_1)$
    \item for any $L_1$ $\p$-isometric to $L$ the set $|\{ L' : L'~\p\text{-isometric to }L , |\bv(L',L_1)|_1 =1  \}| < \infty$ 
    \item for any $L_1,L_2$ $\p$-isometric to $L$, for any unitary $U$: $\bv(U L_1,U L_2) = \bv(L_1, L_2)$
    \item for any $L_1,L_2,L_3$ $\p$-isometric to $L$: $\bv(L_1,L_2) \wmaj  \bv(L_1,L_3) +  \bv(L_3,L_2)$
\end{itemize}
\end{definition}

Or goal is to show the following property for a range of Hermitian lattices $L$ and sets of unitaries $G$ that is more general than \cref{eq:lmaj-reduction}
\begin{equation}
\label{eq:lmaj-lattice-reduction}
\text{For any}~L_1~\p-\text{isometric to }L, L \ne L_1 \text{ there is }g\in G : \bv(L,g L_1) \lmaj \bv(L,L_1)
\end{equation}
\cref{eq:lmaj-reduction} follows from the above property by considering $L_1 = U^\dagger L$.
\cref{eq:lmaj-lattice-reduction} is still related to infinitely many lattices $\p$-isometric to $L$.
The following property will help up to reduce the question to a finitely many lattices $\p$-isometric to $L$
and justifies calling $\bv$ a coordinate map:
\begin{definition}[Intermediate lattice property]
\label{def:intermediate-lattice}
Let $L$ be Hermitian lattice in $E^{2N}$ and $\bv$ a coordinate map~(\cref{def:coordinate-map}) defined on lattices $\p$-isometric to $L$,
we say that the pair $L,\bv$ has an \textbf{intermediate lattice property} if for any 
lattices $L_1,L_2$ $\p$-isometric to $L$ and non-negative vectors $x,y$ such that $x + y = \bv(L_1,L_2)$ there exist 
lattice $L_3$ $\p$-isometric to $L$ such that $\bv(L_1,L_3) = x^\downarrow$, $\bv(L_2,L_3) = y^\downarrow$,
where $x^\downarrow, y^\downarrow$ are vectors $x,y$ sorted in non-increasing order.
\end{definition}

\begin{table}[ht]
    \centering
    \begin{tabular}{|l|c|c|c|c|c|}
    \hline 
    Basis \textbf{$B$} & $2N$ & Field $E$ & Ideal $\mathfrak{p}$ & $m$ & $n$\tabularnewline
    \hline 
    \hline 
    $B_{\mathbb{C}}^{\otimes n}$ (\cref{eq:basis-changes}) & $2^{n}$ & $\mathbb{Q}(\zeta_{2^{m}})$ & $(1+\zeta_{2^{m}})\mathcal{O}_{E}$ & $2$ & $2,\ldots,5$\tabularnewline
    \cline{5-6} \cline{6-6} 
    &  &  &  & $3,4,5$ & $1,\ldots,5$\tabularnewline
    \hline 
    $B_{\mathbb{R}}^{\otimes n}$  (\cref{eq:basis-changes}) & $2^{n}$ & $\mathbb{Q}(\zeta_{2^{m}})\cap\mathbb{R}$ & $(2+2\cos(\frac{2\pi}{2^{m}}))\mathcal{O}_{E}$ & $3,\ldots,6$ & $2,\ldots,5$\tabularnewline
    \hline 
    $B_{\mathbb{Q},n}$ (\cref{eq:rational-basis})  & $2^{n}$ & $\mathbb{Q}$ & $2\mathbb{Z}$ & $-$ & $3,5,7,9$\tabularnewline
    \hline 
    \end{tabular}
    \caption{Hermitian lattices with basis $B$ in $E^{2N}$ that admit a coordinate map~(\cref{def:coordinate-map}) with an intermediate lattice property~(\cref{def:intermediate-lattice})
    for the values of $m$, $n$ listed in the last two columns. 
    See \cref{eq:coord-map-via-valuation} for the definition of the coordinate map for lattices with basis $B$.}
    \label{tab:lattices-with-intermediate-property}
\end{table}

We call lattices $L_1,L_2$ such that $l_1$-norm $|\bv(L_1,L_2)|_1 = 1$ \textbf{neighbours}.
Intermediate lattice property implies that all lattices $\p$-isometric to $L$ form a connected graph
with any two lattices $L_1,L_2$ connected by a path of length $|\bv(L_1,L_2)|_1$.
The following theorem summarises which lattices admit a coordinate map with the intermediate lattice property.

\begin{theorem}[Lattices with an intermediate property]
\label{thm:lattices-with-intermediate-property}
Consider Hermitian lattice $L$ with basis $B$ in $E^{2N}$ and the prime $\Oe$-ideal $\p$ with norm two, $B,E,\p$ from \cref{tab:lattices-with-intermediate-property}.
Define map $\bv$ on pairs of Hermitian lattices $\p$-isometric to $L$ as following. 
For $L_1, L_2$  $\p$-isometric to $L$ and the invariant factors $\mathfrak{r}_1,\ldots, \mathfrak{r}_{2N}$ of $L_1$ in $L_2$, define:
\begin{equation}
\label{eq:coord-map-via-valuation}
\bv(L_1,L_2) = (-v_\p(\mathfrak{r}_1), \ldots, -v_\p(\mathfrak{r}_N) ).
\end{equation}
Then, the map $\bv(L_1,L_2)$ is a coordinate map~(\cref{def:coordinate-map}) and pair $L,\bv$ has an intermediate lattice property~(\cref{def:intermediate-lattice}).
When $L_1$ and $L_2$ are isometric to $L$ with isometries $U_1,U_2$, $\bv(L_1,L_2) = \nuv_B(U_2^\dagger U_1)$~(see \cref{thm:unitary-invariant-factors} for $\nuv_B$).
\end{theorem} 

We provide the proof of \cref{thm:lattices-with-intermediate-property} in the end of \cref{sec:lattices-with-intermediate-property}.
There are efficient algorithms for computing the invariant factors of $L_1$ in $L_2$ and, consequently, coordinate map $\bv(L_1,L_2)$ defined in the theorem.
In particular, when $\Oe$ is a principal ideal domain, lattices $L_1,L_2$ have bases $B_1, B_2$
and the invariant factors can be found by using a Smith Normal Form algorithm on $B_1^{-1} B_2$.

It is easier to establish~\cref{thm:best-first-search} when genus class number of $L = B \Oe^{2N}$ is one.
Interestingly, in this case the problem graph introduced in \cref{sec:problem-graph} with appropriately chosen set of generators is isomorphic to the graph of $\p$-isometric lattices.
This is the case for lattices in first two rows of \cref{tab:best-first-search}. 
In~\cref{tab:genus-class-number} we list the lattices and their genus class number that we have computed using Hecke~\cite{Nemo}.
The other two lattices in~\cref{thm:best-first-search} have genus class number three 
and will require a more sophisticated proof technique. 

\begin{proof}[Proof of \cref{thm:best-first-search}, genus class one case]
We first check that the following two sets are equal 
\begin{equation}
\label{eq:gates-and-neighbours}
    \{ g L  : g \in G \} = \{ L' : |\bv(L',L)|=1 : L'~\p-\text{isomertric to }L \}.
\end{equation}
for $L = B \Oe^{2N}$ and set $G$ described in the theorem and corresponding to the \textbf{ first two rows} of \cref{tab:best-first-search}.
We provide algorithm~\cref{alg:neighbours} to compute the set in the right-hand size in \cref{sec:lattice-graph}. 
We have used Magma~\cite{Magma} and Hecke~\cite{Nemo} to check that \cref{eq:gates-and-neighbours} 
holds for the first two rows in \cref{tab:best-first-search}. 

Consider now $L_1 = V^\dagger L$ for $V$ in \cref{line:lmaj-reduction} in \cref{alg:best-first-search}, 
and let $z = \bv(L, L_1) = \nuv_B(V) = \nuv_B(V^\dagger)$. 
Let $j$ be the index of the last non-zero coordinate of $z$ and let us apply the intermediate lattice property for $z = x  + y$
where coordinates $x_i = \delta_{i,j}$.
Our choice of $x$ ensures that $y$ is non-increasing non-negative integer vector.
There is  lattice $L_2$ such that $\bv(L_2, L_1) = y \lmaj z$ and $|\bv(L,L_2)|_1 = 1$.
Using \cref{eq:gates-and-neighbours} we write $L_2 = g^\dagger L$ for some $g$ from $G$,
and we have 
$$
\nuv_B(V g^\dagger) =  \bv(L, g V^\dagger L) = \bv(g^\dagger L, V^\dagger L) = \bv(L_2, L_1) = x \lmaj \bv(L, L_1) = \nuv_B(V).
$$
Every step of \cref{alg:best-first-search} reduces the sum 
$$
\sum_{k=1}^N \sum_{j=1}^k (\nuv_B(V))_j
$$
at least by one, and this sum is upper-bounded by a constant times $\max \nuv_B(U)$ in the beginning of the algorithm execution.
This show that the algorithm terminates in the number of steps linear in $\max \nuv_B(U)$.
\end{proof}

The gate set normalization~(\cref{sec:normalized-gate-sets-algo}) is closely related to \cref{eq:gates-and-neighbours}.
The normalized gate sets~(\cref{def:normalized-gate-set}) $G$ have a property that lattices $gL$ for $g$ from $G$ are all distinct. 
To handle genus class number greater than one we require an additional definition.

 




\begin{table}[ht]
    \centering
\begin{tabular}{|c|c|c|c|}
\hline 
Basis $B$ & Field $E$ & Genus class number & Number of neighbours\tabularnewline
\hline 
\hline 
$B_{\mathbb{C}}$ & $\mathbb{Q}(\zeta_{8})$ & $1$ & $3$\tabularnewline
\hline 
$B_{(3)}$ & $\mathbb{Q}(\zeta_{3})$ & $1$ & $12^\ast$ \tabularnewline
\hline 
$B_{\mathbb{C}}^{\otimes2}$ & $\mathbb{Q}(i)$ & $1$ & $15$ \tabularnewline
\hline 
$B_{\mathbb{R}}^{\otimes2}$ & $\mathbb{Q}(\sqrt{2})$ & $1$ & $9$ \tabularnewline
\hline 
$B_{\mathbb{Q},3}$ & $\mathbb{Q}$ & $1$ & $135$ \tabularnewline
\hline 
\hline 
$B_{\mathbb{C}}$ & $\mathbb{Q}(\zeta_{16})$ & $2$ & $3$ \tabularnewline
\hline 
$B_{(3)}$ & $\mathbb{Q}(\zeta_{9})$ & $3$ & $20^\ast$ \tabularnewline
\hline 
$B_{\mathbb{C}}^{\otimes2}$ & $\mathbb{Q}(\zeta_{8})$ & $3$ & $15$ \tabularnewline
\hline 
$B_{\mathbb{C}}^{\otimes3}$ & $\mathbb{Q}(i)$ & $3$ & $270$ \tabularnewline
\hline 
$B_{\mathbb{R}}^{\otimes3}$ & $\mathbb{Q}(\sqrt{2})$ & $6$ & $135$ \tabularnewline
\hline 
$B_{\mathbb{C}}$ & $\mathbb{Q}(\zeta_{32})$ & $58$ & 3 \tabularnewline
\hline 
$B_{\mathbb{C}}^{\otimes 2}$ & $\mathbb{Q}(\zeta_{16})$ & $> 174$ & $15$\tabularnewline
\hline 
$B_{\mathbb{C}}^{\otimes3}$ & $\mathbb{Q}(\zeta_8)$ & $> 287$ & $270$\tabularnewline
\hline 
\end{tabular}
    \caption{Genus class number of hermitian lattices $L = B \O_E^{N}$ with basis $B$ in $E^N$.
    Matrices $B_{\mathbb{C}}$, $B_{\mathbb{R}}$, $B_{\mathbb{Q},n}$, $B_{(3)}$ are defined in \cref{eq:basis-changes}, \cref{eq:rational-basis}.  }
    \label{tab:genus-class-number}
\end{table}

\begin{definition}[Reduction property]
\label{def:reduction}
Let $L$ be Hermitian lattice in $E^{2N}$, let $\bv$ be a coordinate map~(\cref{def:coordinate-map}) associated with $L$, let $x$ be non-negative non-increasing integer vector and let $G$ be a set of unitaries over $\xi$-ring $\R \subset E$.
We say that tuple $(L,\bv,x, G)$ has a \textbf{reduction property} if for any lattice $L_1$ $\p$-isometric to $L$ with $\bv(L,L_1)  = x$ 
there exists $U$ from $G$ such that $\bv(L, U L_1) \lmaj \bv(L, L_1)$.
We say that the tuple $(L,\bv,x, G)$ has a \textbf{weak reduction property} if the condition holds only for lattices $L_1$ isometric to $L$.
\end{definition}
Note that if there are no lattices $L_1$ isometric to $L$ such that $\bv(L,L_1)  = x$, then the weak reduction property holds trivially.
The reduction property is crucial for reducing \cref{eq:lmaj-lattice-reduction} to checking a finite number of cases for which the reduction or weak reduction property holds.
First, the number of lattices $L_1$ such that $\bv(L,L_1)  = x$ is finite because the intermediate lattice property and the fact that the number of neighbours of each lattice is finite.
Second, the following lemma shows that reduction property extends from one vector $x$ to infinitely many related vectors: 

\begin{lemma}[Coordinate reduction]
\label{lem:reduction}
\textbf{Consider} a Hermitian lattice in $E^{2N}$, 
a coordinate map~(\cref{def:coordinate-map}) $\bv$ associated with $L$, non-negative non-decreasing integer vector $x$
and a set $G$ of unitaries over $\xi$-ring $\R \subset E$ such that
\begin{itemize}
    \item $L,\bv$ have intermediate lattice property~(\cref{def:intermediate-lattice}),
    \item $(L,\bv,G,x)$ has the reduction property~(\cref{def:reduction}),
\end{itemize}
\textbf{then} for any non-negative non-decreasing $y$ tuple $(L,\bv,G, x+y)$ has the reduction property.
\end{lemma}
\begin{proof}
Consider lattice $L_1$ $\p$-isometric to $L$ such that $\bv(L,L_1) = x + y$.
Using intermediate lattice property~(\cref{def:intermediate-lattice}) there exist lattice $L_2$ $\p$-isometric to $L$ such that $\bv(L,L_2) = x$ and $\bv(L_2,L_1) = y$.
Next, for any unitary $U'$ with entries in $\R$ we have 
\begin{align}
\label{eq:reduction-bound}
\bv(L , U' L_1 ) 
&\wmaj
\bv(L , U' L_2) + \bv(U' L_2, U' L_1) =
\\ 
& =
\bv(L , U' L'_2) + \bv(L_2, L_1) =
\\
& =
\bv(L , U' L'_2) - \bv(L, L_2) + \bv(L, L_2) + \bv(L_2, L_1)
\\
& = 
(\bv(L , U' L_2) - \bv(L, L_2)) + \bv(L, L_1)
\end{align}
Using reduction property~(\cref{def:reduction}), there exist a choice of $U'$ that we denote by $U$ so that 
$$
\bv(L , U L_2) \lmaj \bv(L, L_2) 
$$
For this choice of $U$, using \cref{prop:maj-relations}, we have 
$
\bv(L , U L_1) \lmaj \bv(L, L_1)
$,
which shows the required result.
\end{proof}

The reduction property and weak reduction property can be numerically checked by enumerating all the lattices $L_1$ 
such that $\bv(L,L_1) = x$. 
Such enumeration process can be compress by using automorphism group of $L$. 
Let $C$ be a unitary such that $C L = L$ and suppose that $\bv(L, U L_1) \lmaj \bv(L, L_1)$ for some $U$.
Using the properties of the coordinate map $\bv$ we have 
$$
\bv(L, (C_1 U C^\dagger) C L_1) = \bv(C_1 ^\dagger L, U L_1) = \bv( L, U L_1) \lmaj \bv(L, L_1) = \bv(C^\dagger L, L_1) = \bv( L,C L_1)
$$
When gate set $G$ is normalized with the respect to automorphism group of $L$, 
it is sufficient to check the reduction property for one element of the set
\begin{equation}
\label{eq:orbit}
    \{ C L_1 : C \text{ is an automorphism of } L \}.
\end{equation}
We call above set an \textbf{orbit} of $L_1$ under automorphism group of $L$. 
When discussing the reduction property we simply refer to these sets as orbits.
We have checked the reduction property for the range of values of $x$ for the 
lattices listed in \cref{tab:best-first-search}. 
The results of our computations are summarized in \cref{tab:ct-reduction}, \cref{tab:cty-reduction}.
We next show that first two rows of the tables suffice to establish reduction property for all but finitely many lattices $\p$-isometric to $L$. 
We need an additional definition to make arguments simpler.

\begin{table}[ht]
    \centering
\begin{tabular}{|c|c|c|c|c|}
\hline 
$x$ & $\partial x$ & Reduction & Weak reduction & Orbits\tabularnewline
 &  & property & property & per class\tabularnewline
\hline 
\hline 
$(2,0)$ & $(2,0)$ & $+$ & + & $[2,1,1]$\tabularnewline
\hline 
$(2,2)$ & $(0,2)$ & $+$ & $+$ & $[3,1,2]$\tabularnewline
\hline 
\hline 
$(2,1)$ & $(1,1)$ & $+$ & $+$ & $[0,2,1]$\tabularnewline
\hline 
$(1,1)$ & $(0,1)$ & $-$ & $+$ & $[1,0,1]$\tabularnewline
\hline 
$(1,0)$ & $(1,0)$ & $-$ & $+$ & $[0,1,0]$\tabularnewline
\hline 
\end{tabular}
    \caption{Lattices in $\q(\zeta_8)^4$ $\p$-isometric to $B_{\mathbb{C}}^{\otimes 2}$ with reduction property~(\cref{def:reduction}) with respect to normalized Clifford and $\text{T}$, $\text{T}\otimes \text{T}$, $\text{CT}$ gate set.
    Prime ideal $\p = (1+\zeta_8) \mathbb{Z}[\zeta_8]$.
    We also list the number of orbits~(\cref{eq:orbit}) of lattices with coordinate $x$ per isometry class. }
    \label{tab:ct-reduction}
\end{table}

\begin{table}[ht]
    \centering
\begin{tabular}{|c|c|c|c|c|}
\hline 
$x$ & $\partial x$ & Reduction & Weak reduction & Orbits\tabularnewline
 &  & property & property & per class\tabularnewline
\hline 
\hline 
$(2,0)$ & $(2,0)$ & $+$ & + & $[1,1,1]$\tabularnewline
\hline 
$(2,2)$ & $(0,2)$ & $+$ & $+$ & $[1,0,2]$\tabularnewline
\hline 
\hline 
$(2,1)$ & $(1,1)$ & $+$ & $+$ & $[0,1,1]$\tabularnewline
\hline 
$(1,1)$ & $(0,1)$ & $-$ & $+$ & $[1,0,1]$\tabularnewline
\hline 
$(1,0)$ & $(1,0)$ & $-$ & $+$ & $[0,1,0]$\tabularnewline
\hline 
\end{tabular}
    \caption{Lattices  in $E^4 = \q(\cos(\pi/8))^4$ $\p$-isometric to $B_{\mathbb{R}}^{\otimes 2}$ with reduction property with respect to normalized Clifford and $\text{T}_y$, $\text{CS}_y$, $\text{CT}_y$ gate set.
    Prime ideal $\p = (2+2\cos(\pi/8)) \Oe$.
    We also list the number of orbits~(\cref{eq:orbit}) of lattices with coordinate $x$ per isometry class.}
    \label{tab:cty-reduction}
\end{table}

For non-increasing non-negative $N$-dimensional vector $x$ define derivative vector
$$
\partial x = (x_1 - x_2, \ldots, x_{N-1} - x_N, x_N)
$$
The $\partial x$ is non-negative, moreover for any non-negative vector $y$ we have 
$$
y = \partial\left( \sum_{j = 1}^{N} y_j v_j \right), \text{ where } v_j = (\underbrace{1,\ldots,1}_{j},\underbrace{0,\ldots,0}_{N-j})
$$ 
Using above notation the intermediate lattice property implies that for any 
$L_1,L_2$ and any non-negative integer vectors $x,y$
such that 
$$
\partial \nuv(L_1,L_2) = x' + y'
$$
there exist a lattice $L_3$ such that 
$$
\partial \nuv(L_1,L_3) = x' ,~~\partial \nuv(L_2,L_3) = y'
$$

Informally, the following lemma shows that if the reduction property holds for  $(L,\bv,G,a_j v_j)$ for some $a_j$, 
than for all lattices $L_1$ except a finite set there exist a unitary $U$ such that $\bv(L,UL_1) \lmaj \bv(L,L_1)$.
\begin{lemma}[Box Lemma]
\label{lem:box-lemma}
Consider a Hermitian lattice in $E^{2N}$, 
a coordinate map~(\cref{def:coordinate-map}) $\bv$ associated with $L$ with values in $\mathbb{Z}^{N'}$ and a set $G$ of unitaries over $\xi$-ring $\R \subset E$ such that
\begin{itemize}
    \item $L,\bv$ have intermediate lattice property~(\cref{def:intermediate-lattice}),
    \item there exist positive integers $a_1,\ldots,a_{N'}$ such that $(L,\bv,G, a_j v_j)$ have reduction property~(\cref{def:reduction}), where $v_j = (\underbrace{1,\ldots,1}_{j},\underbrace{0,\ldots,0}_{N'-j})$
\end{itemize}
Then for any lattice $L_1$ $\p$-isometric to $L$ with 
$$
 \partial \bv(L, L_1) \notin [0,a_1) \times \ldots \times [0,a_{N'})
$$
there exist $U$ from $G$ such that $\bv(L, U L_1) \lmaj  \bv(L, L_1)$.
\end{lemma}
\begin{proof}
If the condition on $L_1$ holds it means that there is $j$ such that $\partial \bv(L, L_1)_j \ge a_j$ 
and therefore $\bv(L, L_1) = a_j v_j + y$ for some non-negative non-decreasing integer vector $y$. 
The result follows from~\cref{lem:reduction}.
\end{proof}
The set of non-negative vectors not in $ [0,a_1) \times \ldots \times [0,a_{N'})$ is finite and therefore the related set of lattices is also finite.

\begin{proof}[Proof of \cref{thm:best-first-search}, general case]
For the lattices $L$ related to the rows $3,4$ in \cref{tab:best-first-search} via $L = B \Oe^{2N}$,
the reductions property holds for vectors $x = (2,2), y = (2,0)$ with 
$\partial x = (0,2), \partial y = (2,0)$ based on the exhaustive computer search, summarized in \cref{tab:ct-reduction}, \cref{tab:cty-reduction}.
\cref{lem:box-lemma} implies that all the lattices $L_1$ with $\partial \bv(L,L_1) \notin [0,2) \times [0,2) $
have reductions property. 
The rest of the lattices have the weak reduction property as summarized in  \cref{tab:ct-reduction}, \cref{tab:cty-reduction}.
Weak reduction property hold for all lattices $\p$-isometric to $L$ and therefore \cref{eq:lmaj-reduction} holds.

As before, every step of \cref{alg:best-first-search} reduces the sum 
$$
\sum_{j=1}^N \sum_{k=1}^j (\nuv_B(V))_j
$$
at least by one, and this sum is upper bounded by a constant times $\max \nuv_B(U)$ in the beginning of the algorithm execution.
This show that the algorithm terminates in the number of steps linear in $\max \nuv_B(U)$.
\end{proof} 

\subsection{Lattices with the intermediate lattice property}
\label{sec:lattices-with-intermediate-property}

The goal of this subsection is to prove \cref{thm:lattices-with-intermediate-property}. 
We start with the key lemma.

\begin{lemma}[Local diagonalization]
\label{lem:local-diag}
Consider $L_\p$, a scaled hyperbolic Hermitian lattice in $E^{2N}_\p$
with hyperbolic basis $e_1,\ldots,e_N,f_1,\ldots,f_N$ and scale $\alpha$.
For any unitary $\lU$ with entries in $E_\p$ there exist unitaries $\lU_L$, $\lU_R$ that preserve $L_\p$
and unitary 
\begin{equation}
\label{eq:local-diag}
\lD = \sum_{i = 1}^N a_i e_i f_i^\dagger + (a_i^\ast)^{-1} f_i e_i^\dagger, |v_\p(a_i)| \ge |v_\p(a_{i+1})|
\end{equation}
diagonal in the hyperbolic basis such that $\lU = \lU_L \lD \lU_R$.
\end{lemma}
\begin{proof}
Our proof strategy is to gradually transform $\lU$ into $\lD$ by using unitaries that preserve 
$
L_\p
$.
The first step is to describe the structure of unitaries that preserve $L_\p$ in a more detailed way.

We will write $2N \times 2N$ matrices with entries in $E_\p$ 
with respect to the basis $e_1,\ldots,e_N$, $f_1,\ldots,f_N$ in block form
\begin{equation*} 
\lU = 
\begin{pmatrix}
A & \nicefrac{1}{\alpha}\,  B \\ \alpha^\ast C & D
\end{pmatrix},
~~
\lU^{-1} = 
\begin{pmatrix}
D^\dagger &  \nicefrac{1}{\alpha^\ast }\,B^\dagger \\ \alpha C^\dagger & A^\dagger
\end{pmatrix},
\end{equation*}
with $A,B,C,D$ being $N \times N$ matrices with entries in $E_\p$.
Matrix $\lU$ is a unitary if and only if $\langle \lU e_i, \lU e_j \rangle = 0$, $\langle \lU f_i, \lU f_j \rangle = 0$, 
$\langle \lU e_i, \lU f_j \rangle = \delta_{i,j}$. 
In terms of matrices $A,B,C,D$ these conditions are:
$u\,C^\dagger A+A^\dagger C=u^\ast \,D^\dagger B+B^\dagger D=0$ and $A^\dagger D+C^\dagger B=I_{N}$, where $u = \alpha / \alpha^\ast$. 
Moreover, using the expression for $\lU^{-1}$, we see that $\lU$ preserves $L_\p$ if and only if entries of $A,B,C,D$ are in $\O_\p$. 

In particular the following unitary matrices preserve $L_\p$:
\begin{equation*}
\begin{split}
\mathrm{diag}(A,A^{-\dagger})~,~& A \text{ is } \O_\p\text{-unimodular} \\
\begin{pmatrix}
0 & \nicefrac{1}{\alpha}\, B \\
\alpha^\ast B^{-\dagger} & 0
\end{pmatrix},~ &~ B \text{ is } \O_\p\text{-unimodular}  \\
\begin{pmatrix}
I_N & 0 \\
\alpha^\ast C & I_N
\end{pmatrix},~ &~ C \text{ is } \O_\p\text{-unimodular and }  C^\dagger= - u^\ast \,C, \\
\end{split}
\end{equation*}
notation $A^{-\dagger}$ is for $(A^{-1})^\dagger$, for the term $\O_\p$-unimodular see \cref{def:xi-modular}.
Above implies that we can do arbitrary permutations in the $e_i$ (if we also do the same permutation on the $f_i$) and  
we can exchange blocks $(A,B)$ for $(C,D)$, $(C,D)$ for $(A,u^\ast B)$.

Next we proceed to gradually transform $\lU$ into a diagonal matrix.
We can always find $L_\p$ preserving unitaries $\lU_{L,0}, \lU_{R,0}$ such that $\lU_1 = \lU_{L,0} \lU \lU_{R,0}$ has top-left entry 
with the minimal valuation $v_\p$ among all the entries of sub-matrices $A$, $B$, $C$, $D$ associated with $\lU_1$.
Let us now denote the first column of $\lU_1$ in the hyperbolic basis by $(a_1,\ldots,a_N,\alpha^\ast c_1,\ldots, \alpha^\ast c_N)^{T}$. 
Set
\begin{equation*}
A_1 =\begin{pmatrix}
1 & 0 & ... & 0\\
-\frac{a_2}{a_1} & 1 &...& 0\\
\vdots & 0 &\ddots&  0\\
-\frac{a_N}{a_1}& 0 &  ...& 1
\end{pmatrix},
\end{equation*}
and $\lU_{L,1}:=\mathrm{diag}(A_1,A_1^{-\dagger})$.
By our assumption on the valuation of $a_1$ unitary $\lU_{L,1}$ preserves $L_\p$ 
and moreover the first column $x_1$ of $\lU_{L,1} \lU_1$ in the hyperbolic basis is of the form
$(a_1,0,\ldots,0,\alpha^\ast c'_1,\ldots,\alpha^\ast c'_N)$.
Since $\lU_{L,1} \lU_1$ is unitary we have
\begin{equation*}
0= \langle x_1, x_1 \rangle =a_1 (\alpha^\ast c'_1)^\ast + a_1^\ast \alpha^\ast c'_1
\end{equation*}
and therefore 
\begin{equation}
\label{eq:top-left-entry}
-\frac{\alpha^\ast}{\alpha} \frac{c_1'}{a_1} + \left(-\frac{c_1'}{a_1}\right)^\ast = -\frac{a_1 (\alpha^\ast c'_1)^\ast + a_1^\ast \alpha^\ast c'_1}{\alpha a_1 a_1^\ast} = 0.
\end{equation}
We set
\begin{equation*}
C_1:=\begin{pmatrix}
-\frac{c'_1}{a_1} & u\left(\frac{c'_2}{a_1}\right)^\ast& \ldots & u\left(\frac{c'_N}{a_1}\right)^\ast\\
-\frac{c'_2}{a_1} & 0 & \ldots & 0\\
\vdots& \vdots & \ddots& \vdots \\
-\frac{c'_N}{a_1} & 0 & \ldots & 0
\end{pmatrix}.
\end{equation*}
Matrix $C_1$ is $\O_\p$-unimodular because we chose $a_1$ with the smallest valuation among all entries of $A,B,C,D$ associated with $\lU$ and ratios $c'_j / a_1$ are $\O_\p$-integer combinations of entries of sub-matrices $A$, $B$, $C$, $D$ associated with $\lU$ divided by $a_1$.
Matrix $C_1^\dagger = -u^\ast C_1$~(\cref{eq:top-left-entry}) and for $L_\p$-preserving unitary
\begin{equation*}
\lU_{L,2}:=\begin{pmatrix}
I_m & 0 \\ \alpha^\ast C_1 & I_m
\end{pmatrix}
\end{equation*}
we see that the first column of $\lU_{L,2}\lU_{L,1}\lU_{L,0} \lU $ is $(a_1,0,...,0)$.

Performing the analogous operation from the right,
we find $L_\p$-preserving unitaries $\lU_{R,0},\lU_{R,1},\lU_{R,2}$ that such that $\lU_{L,2}\lU_{L,1}\lU_{L,0}\,\lU\,\lU_{R,0}\lU_{R,1}\lU_{R,1}$ has first column and row equal to $(a_1,0,...,0)$.
	
By induction this means we can find $L_\p$-preserving unitaries  $\lU'_L,\lU'_R $ such that 
\begin{equation*}
\lU'_L \, \lU \, \lU'_R=\begin{pmatrix}
A & 0 \\0 & A^{-\dagger}
\end{pmatrix},
\end{equation*}
where $A$ (and thus also $A^{-\dagger}$) is a diagonal matrix. 
Finally, we can reorder diagonal of $A$ so that the conditions on the valuations in \cref{eq:local-diag} hold. 
\end{proof}

Above lemma implies that local lattices can be expressed in a common hyperbolic basis, that we are going to use later to construct the intermediate lattice.

\begin{lemma}[Common hyperbolic basis]
\label{lem:common-hyperbolic-basis}
Consider two scaled hyperbolic lattices $L_\p$, $M_\p$ in $E_\p^{2N}$~(\cref{def:hyperbolic-lattice})
with scale $\alpha$, then 
there exist common hyperbolic basis $e_1,\ldots,e_N,f_1,\ldots,f_N$ and 
integer vector $(v_1, \ldots, v_N)$ such that 
$$
L_\p = \Op e_1  + \Op \alpha f_1   + \ldots + \Op e_N  + \Op \alpha f_N , ~~  M_\p = \p^{v_1} e_1  + \p^{-v_1} \alpha f_1 + \ldots + \p^{v_N} e_N  + 
\p^{-v_N} \alpha f_N,
$$
and vector 
\begin{equation}
\label{eq:relative-local-coordinates}
    \nuv(L_\p,M_\p)= (|v_1|, \ldots, |v_N|)
\end{equation}is non-increasing.
Additionally, $\p^{-|v_N|}, \ldots, \p^{-|v_1|}$, $\p^{|v_1|}, \ldots, \p^{|v_N|}$ are the invariant factors of $M_\p$ in $L_\p$~(\cref{thm:invariant-factors}).
\end{lemma}
\begin{proof}
Lattices $L_\p$, $M_\p$ are isometric because they are both scaled hyperbolic with the isometry mapping hyperbolic basis of $L_\p$
to the hyperbolic basis of $M_\p$.
There exist unitary $\U$ such that $M_\p = \U L_\p$.
Let $\{ \tilde e_j, \tilde f_j \}_{j=1}^N$ be a hyperbolic basis of $L_\p$.
By the local diagonalization~\cref{lem:local-diag} there exist 
unitaries $\U_L, \U_R$ that preserve $L_\p$ and unitary $\mathfrak{D}$
diagonal in the hyperbolic basis  $\{ \tilde e_j, \tilde f_j \}_{j=1}^N$ such that
$$
\lD = \sum_{i = 1}^N a_i \tilde e_i \tilde f_i^\dagger + (a_i^\ast)^{-1} \tilde f_i \tilde e_i^\dagger,~|v_\p(a_i)| \ge |v_\p(a_{i+1})|~, \U = \U_L \mathfrak{D} \U_R 
$$
We show that $\{ e_j,f_j \}_{j=1}^N = \{ \U_L \tilde e_j, \U_L \tilde f_j \}_{j=1}^N$ is a common hyperbolic basis of $L_\p$, $M_\p$. 
Clearly, $\{ e_j,\alpha f_j \}_{j=1}^N$ is a basis of $L_\p$ because $\U_L$ preserves $L_\p$.
Using that $\U_R L_\p = L_\p$ we have 
$$
M_\p = \sum_{j=1}^N \Op a_j \U_L \tilde e_i + \Op (a_j^\ast)^{-1} \U_L \alpha \tilde f_j.
$$
Using notation $v_j = v_\p(a_j \Op)$ and observing 
that $a_j \Op = \p^{v_j}$ and $(a_j^\ast)^{-1} \Op = \p^{-v_j}$ we get the required 
expression for $M_\p$. 
Finally, using the fact that $\p^n \subseteq \p^m$ when $n \ge m$
and the uniqueness of invariant factors~(\cref{thm:invariant-factors}) we get the required result.
\end{proof}

Next we use the common hyperbolic basis to construct a local intermediate lattice.

\begin{corollary}[Local intermediate lattice]
\label{cor:local-intermediate-lattice}
Consider scaled hyperbolic lattices $L_\p$, $M_\p$ in $E_\p^{2N}$ with scale $\alpha$, 
then for any non-negative integer vectors $x,y$ such that $ x + y = \nuv(L_\p,M_\p) $ there exist a hyperbolic lattice $\tilde L_\p$ in $E_\p^{2N}$,
such that 
$
 \nuv(L_\p,\tilde L_\p) = x^\downarrow, ~ \nuv(M_\p,\tilde L_\p) = y^{\downarrow}
$, where $x^\downarrow, y^\downarrow$ are vectors $x,y$ sorted in non-decreasing order.
\end{corollary}
\begin{proof}
Using common hyperbolic basis $\{ e_j, f_j \}_{j=1}^N$ of $L_\p$, $M_\p$~(\cref{lem:common-hyperbolic-basis}) 
we can write these lattices as 
$$
L_\p = \Op e_1  + \Op \alpha f_1   + \ldots + \Op e_N  + \Op \alpha  f_N ,
~~
M_\p = \p^{v_1} e_1  + \p^{-v_1} \alpha  f_1 + \ldots + \p^{v_N} e_N  + 
\p^{-v_N} \alpha f_N,
$$
where $(|v_1|,\ldots,|v_N|) = \nuv(L_\p,M_\p)$.
Define 
$$
\tilde L_\p = \p^{x_1 s(v_1)} e_1  + \p^{-x_1 s(v_1)} \alpha f_1 + \ldots + \p^{x_N s(v_N)} e_N  + 
\p^{-x_N s(v_N)} \alpha f_N.
$$
where $s(a) = 1$ if $a \ge 0$, $s(a) = -1$ if $a < 0 $.
We have $\nuv(L_\p, \tilde L_\p) = x^{\downarrow}$, $\nuv(M_\p, \tilde L_\p) = y^{\downarrow}$.
\end{proof}

The next result show how to go from local intermediate lattices to global ones.

\begin{lemma}[Intermediate lattice]
\label{lem:intermediate-lattice}
Consider either a CM field or a totally real field $E$ and let $\p$ be a prime $\Oe$-ideal.
Additionally, consider two Hermitian lattices $L$ and $M$ in $E^{2N}$ such that:
\begin{itemize}
    \item lattice $L$ and $M$ are $\p$-isometric~(\cref{def:p-isometric})
    \item $L_\p, M_\p$ are scaled hyperbolic lattices~(\cref{def:hyperbolic-space}),
\end{itemize}
The invariant factors of $M$ in $L$ are $(\p^{v_1},\ldots,\p^{v_N},\p^{-v_N},\ldots,\p^{v_1})$
for decreasing non-negative integer vector $\nuv(L,M) = (v_1,\ldots,v_N)$.
Moreover, for any two non-negative integer vectors $x,y$ such that $v = x+y$,
there exist lattice $L'$ in $E^{2N}$  $\p$-isometric to $L,M$ such that $\nuv(L,L') = x^\downarrow$ and $\nuv(L',M) = y^\downarrow$.
\end{lemma}
\begin{proof}
Using \cref{prop:invariant-factor-loclization} we see that invariant factors on $M$ in $L$ are determined by the  invariant factors of $M_\p$ in $L_\p$,
that is invariant factors of $M$ in $L$ are $(\p^{k_1},\ldots,\p^{k_{2N}})$ where $\p$ is $\Oe$-ideal,
and invariant factors of $M_\p$ in $L_\p$ are $(\p^{k_1},\ldots,\p^{k_{2N}})$ where $\p$ is $\Op$-ideal.
Using \cref{lem:common-hyperbolic-basis}, we see that the powers of $\p$ in the invariant factors have required properties.
Next we construct $\mathfrak{L}'$ using local intermediate lattice property~(\cref{cor:local-intermediate-lattice}) so that 
$v(\mathfrak{L}',L_\p) = x^\downarrow$ and $v(\mathfrak{L}',M_\p) = y^\downarrow$,
and then construct $L'$ using \cref{thm:local-global-lattice} so that $L'_\p = \mathfrak{L}'$ and $L'_\mathfrak{q}$ is equal to $L_\mathfrak{q},M_\mathfrak{q}$ for all other prime ideals $\mathfrak{q} \ne \p$.
Using the relation between invariant factors of $L'$ in $L,M$ and $L'_\p$ in $L_\p,M_\p$ completes the proof.
\end{proof}

\begin{proof}[Proof of \cref{thm:lattices-with-intermediate-property}]
We first note that $\alpha L_\p$ is a scaled hyperbolic lattice for some scalar $\alpha$ from $E$.
For all the lattices listed in \cref{tab:lattices-with-intermediate-property} this follows from 
\cref{lem:bw-lattice-localization,lem:real-bw-lattice-localization,lem:rational-bw-lattice-localization}.
Because the invariant factors of $L$ in $M$ and of $\alpha L$ in $\alpha M$ are the same, 
we are going to work with the appropriately re-scaled lattices for the rest of the proof.

Let us next show that $\bv$ is indeed a coordinate map.
Let $L_1,L_2,L_3$ be $\p$-isometric to $L$.
Equality $\bv(L_1,L_2) = 0$ implies $L_1 = L_2$ by the definition of $\bv$ and the invariant factors theorem~\cref{thm:invariant-factors}.
Property $\bv(L_1,L_2) = \bv(L_2,L_1)$ follows from the structure of invariant factors of $L_1$ in $L_2$ shown in~\cref{lem:intermediate-lattice},
and expressing the invariant factors of $L_2$ in $L_1$ in terms of the invariant factors of $L_1$ in $L_2$, by definition~\cref{eq:invariant-factors}.
The fact that the number of lattices $L_1'$ with $l_1$-norm $|\bv(L_1,L_2)|_1 = 1$ is finite follows from the correctness proof of \cref{alg:neighbours} in the next \cref{sec:lattice-graph}.
The invariance under multiplication by a unitary  $\bv(U L_1,U L_2) = \bv(L_1,L_2)$ follows from the definition of invariant factors in~\cref{eq:invariant-factors}. 
The inequality $\bv(L_1,L_2) \wmaj \bv(L_1,L_3) + \bv(L_3,L_2)$ follows from the observation that the invariant factors 
and value of $\bv$ are defined by the localizations $L_{j,\p}$ of $L_j$ at $\p$. 
For the localizations, the relation follows from \cref{thm:unitary-invariant-factors}, with $\xi$-ring $\R$ corresponding to $E_\p$.

The intermediate lattice property with respect to $\bv$ follows from \cref{lem:intermediate-lattice}. 
The relation $\nuv_B(U_2^\dagger U_1) = \bv(U_1 L, U_2 L)$ follows from the uniqueness of the invariant factors.
\end{proof}

\subsection{Graphs of \texorpdfstring{$\p$}{p}-isometric lattices}
\label{sec:lattice-graph}

We say that $\p$-isometric~(\cref{def:p-isometric}) lattices $L$ and $M$ are \textbf{neighbours} when the $l_1$-norm $|\nuv(L,M)|_1 = 1$.
There is an algorithm to enumerate all neighbours of a lattice $L$ with $L_\p$ being a scaled hyperbolic lattice. 

\begin{algorithm}[H]
\caption{Neighbours of a lattice $L$}
\label{alg:neighbours}
\begin{algorithmic}[0]
    \State \Input Lattice $L$ in $E^{2N}$ such that $E$ is either CM field or a totally real field, prime $\Oe$-ideal $\p$ with norm $2$; $L$,$\p$ are such that $L_\p$ is a scaled hyperbolic lattice.
    \State \Output All lattices $M$ such that $L,M$ are $\p$-isometric and $\nuv(L,M) = 1$ 
    \State Set $\mathcal{N}$ to an empty list of neighbours 
    \For{$M'$ a maximal sublattice of $L$ with $\p L \subseteq M \subset L$ (\cref{prob:maximal-sublattice})} \label{line:sublattice}
        \For{$M$ a minimal superlattice of $L$ with $M' \subseteq M \subset \p^{-1} M'$ (\cref{prob:minimal-superlattice})}  \label{line:superlattice}
            \If{ $M$ is modular, $\n(M) = \n(L)$, $\s(M) = \s(L)$ (\cref{eq:norm-and-scale}) }
            \State Append $M$ to $\mathcal{N}$
            \EndIf
        \EndFor
    \EndFor
    \State \Return {$\mathcal{N}$}   
\end{algorithmic}
\end{algorithm}

\begin{proof}[Correctness of \cref{alg:neighbours}]
The proof consists of two parts. We first show that all the lattices found by the algorithm are neighbours of $L$.
Second we show that the algorithm finds neighbours of $L$.

First note that the conditions $\p L \subseteq M \subset L$, $M' \subseteq M \subset \p^{-1} M'$ imply that $M_\mathfrak{q}, M'_\mathfrak{q}$ are equal to $L_\mathfrak{q}$
for all prime ideals $\mathfrak{q} \ne \p$. The fact that $M$ is modular, and equality of scale $\s(L) = \s(M)$ and norm $\n(L) = \n(M)$ 
implies that $L_\p$ is isometric to $M_\p$ using the fact that $L_\p$ is a scaled hyperbolic lattice~(\cref{def:hyperbolic-lattice}) and 
\cref{thm:quadratic-isometry}~(when $E$ is totally-real), \cref{thm:hermitian-isometry}~(when $E$ is a CM field).

Second we show for any lattice $M$ that is a neighbour of $L$, there exist a lattice $M'$ such that conditions relating 
$M, M', L$ in \cref{line:sublattice}, \cref{line:superlattice} of~\cref{alg:neighbours} hold.
When $L$ and $M$ are neighbours there exist pseudo-basis $\{ \a_j, b_j \}_{j=1}^{2N}$ of $L$ \cref{lem:intermediate-lattice}
such that 
$$
 M = \p \a_1 b_1 + \a_2 b_2 + \ldots + \a_{2N- 1} b_{2N - 1} + \p^{-1} \a_{2N} b_{2N}.
$$
Using pseudo-basis $\{ \a_j, b_j \}_{j=1}^{2N}$ we define 
$$
M' = \p \a_1 b_1 + \a_2 b_2 + \ldots + \a_{2N- 1} b_{2N - 1} + \a_{2N} b_{2N},
$$
and show that the conditions on $M, M', L$ are satisfied.
We have the following inclusion of lattices: 
$$
\p L \subseteq M' \subseteq L,~ M' \subseteq M \subseteq \p^{-1}M'
$$
It remains to establish that $M'$ is maximal sublattice of $L$ and $M$ is a minimal superlattice of $M'$.
This follows from the fact that the index ideal~(\cref{sec:invariant-factors}) of $M'$ in $L$ and $M'$ in $M$ is a prime ideal~(\cref{prop:maximal-and-prime-index}).

Finally, we note that algorithm terminates because the number of solutions to \cref{prob:maximal-sublattice}, \cref{prob:minimal-superlattice} is finite.
\end{proof}
Similarly to the algorithms for  \cref{prob:maximal-sublattice}, \cref{prob:minimal-superlattice} one can take advantage of the automorphism group of lattice $L$ 
and enumerate only the neighbours not equivalent up to an automorphism of $L$.

One can use \cref{alg:neighbours} to enumerate all lattices $L_1$ such that $\bv(L,L_1) = x$ for coordinate map $\bv$ from \cref{thm:lattices-with-intermediate-property}.
For example, suppose we want to find all $L_{(2,2)}$ such that $\bv(L,L_{(2,2)}) = (2,2)$.
Any  $L_{(2,2)}$ is neighbour of some lattice $L_{(2,1)}$ with $\bv(L,L_{(2,1)}) = (2,1)$.
So it is sufficient to enumerate lattices  $L_{(2,1)}$, their neighbours $L'$ and select the subset with $\bv(L,L') = (2,2)$.
This can be achieved by computing $\bv(L,L')$ via Smith Normal Form~\cref{sec:snf} when the related ring of integers $\Oe$ 
is a principal ideal domain, or using a more general algorithm for computing the invariant factors.

In practice, we try to compress the sets of lattices we enumerate by relying on the automorphism group of $L$. 
For such compression, it can be useful to compute intersection of automorphism groups of two lattices 
or check if two lattices are isomorphic when the isomorphism is restricted to an isomorphism group of $L$.
Both of this problems can be solved with~\cite{Plesken1997} using techniques similar to the techniques for computing 
isomorphism and automorphism groups of Hermitian lattices in Remark~2.4.4 in~\cite{Kirschmer2016}.




\newpage 
\section{Concluding remarks}
\label{sec:conclusion}

There are several future directions and open questions one can explore using the techniques developed here:
$$
~
$$
\begin{itemize}
    \item[] \textbf{Provably efficient state preparation algorithms.} One can consider residues modulo powers of primes in the bases from 
    \cref{tab:basis-changes} and analyze best-first search state preparation by relying on the consistent heuristic from \cref{sec:simple-heuristic}.
    \item[] \textbf{Non-Clifford gate classification. } One can use heuristic function $\nuv_B$ from~\cref{sec:advanced-heuristics} to classify non-Clifford gates. 
    For two qubits we have found that all unitaries with $\nuv_{B}(U) = (1,1)$ are exactly $C (\text{T}\otimes I) C'$ for Clifford unitaries $C$ and $C'$.
    Similarly,  all unitaries with $\nuv_{B}(U) = (2,0)$ are exactly $C (\text{T}\otimes \text{T}) C'$ or $C (\text{CS}) C'$  for Clifford unitaries $C$ and $C'$.
    To make a more fine-grained classification one can consider $\nuv_{B}$ for neighbours of $U$ in the problem graph.
    \item[] \textbf{Synthesis using $\p$-adic calculations.} One can make results on the intermediate lattice property constructive 
    and use them for a different approach to circuit synthesis. Given two Hermitian lattices $L,U L$ for some unitary $U$, one can enumerate all shortest paths 
    between $L$ and $U L$ in the graph of $\p$-isometric lattices and use isometry testing algorithms to decompose $U$ as a product of unitaries
    that could be easier to synthesize using other methods.
    \item[] \textbf{Basis changes that are not modular.} Our results can be extended beyond basis changes that are modular, in other words 
    beyond modular Hermitian lattices. Any lattice over a local ring can be decomposed as a direct sum of modular lattices. 
    There are well known algorithm for computing this decomposition~\cite{Kirschmer2016}, known as the Jordan decomposition.
    The implementation of these algorithms is also available in \cite{Kirschmer2016code}, \cite{Nemo}.
    Using Jordan decomposition one can likely extend our results to the gate sets related to qutrit Clifford groups.
    \item[] \textbf{Computing generators of the group of unitaries with entries in $\xi$-ring $\R$}
    For simplicity, consider the case when there is a modular Hermitian lattice $L$ with basis $\xi$-ring $\R$ and it has the intermediate lattice property.
    One can use the lattice neighbors \cref{alg:neighbours} and isometry testing algorithms 
    to enumerate all lattices $\p$-isometric to lattice $L$ up to isometries. More precisely one can compute a set of isometry classes representatives $L_1,\ldots,L_m$
    and for each of them the set of their $M$ neighbours 
    $$
    N_j = \{ U_{j,1} L_{k_1}, \ldots , U_{j,M} L_{k_{M}} \}, k_j \in \{1,\ldots,m\}, j \in \{1,\ldots,m\}.
    $$
    It is possible to show that any unitary with entries in $\R$ can be expressed as a product of unitaries from $G_L = \{ U_{j,i} : j\in \{1,\ldots,m\}, i \in \{1,\ldots,M\} \}$.
    Using synthesis algorithms from the first part of the paper one can check if some other gate set $G$ generates all the unitaries with entries in $\R$
    by synthesizing all the unitaries from $G_L$ using $G$.
    \item[] \textbf{Provably efficient isometry synthesis algorithms.} Our numerical experiments suggest that best-first search approach to isometry synthesis is 
    efficient, similarly to the best-first search unitary synthesis. Proving the efficiency of isometry synthesis algorithms is an open question.
    \item[] \textbf{Buildings.} The graph of $\p$-isometric lattices is a skeleton of a simplicial complex known as a Bruhat-Tits building.
    For the relation between Hermitian lattices over local rings and building see~\cite{AbramenkoNebeBuildings}. We do not rely on the structure 
    of the buildings. It is possible that using this extra structure can lead to more elegant proofs of the efficiency 
    of the synthesis algorithms, relying less on the enumeration of lattices. The intermediate lattice property can be interpreted as 
    any two lattices having a common apartment in the building.
    \item[] \textbf{Beyond groups with basis change property} A natural generalization of the basis-change property is to consider 
    finite groups that are intersections of automorphism groups of two or more Hermitian lattices. One can use algorithms 
    discussed in the end of \cref{sec:lattice-graph} to construct examples of such groups.
    Extending methods of this paper to this case is an open problem.
    \item[] \textbf{Improved algorithm for lattice enumeration} 
    It might be possible to prove best-first search synthesis efficiency for more gate sets than in \cref{tab:best-first-search} by developing parallel and distribute algorithms for 
    enumerating lattices $L_1$ such that $\bv(L_1,L) = x$ and similarly better algorithms for checking the reduction property.
    \item[] \textbf{Upper and lower-bounds on the T-count and other related metrics} \cref{eq:reduction-bound} in \cref{lem:reduction} can potentially be used to establish  
    upper-bounds on the T-count of the circuits produced by the best-first search \cref{alg:best-first-search}. Heuristic function $\nuv_B(U)$
    and inequality \cref{eq:product-majorization} can be used for the T-count lower-bounds.
\end{itemize}
$$
~
$$
We also envision that our algorithm will lead to new multi-qubit approximate synthesis algorithm, 
as well as better circuits that use catalytic embeddings.
The final, and probably most significant, open question is to how to incorporate the uncomputation 
in multi-qubit synthesis methods~\cite{Jones2013,gidney2021cccz}.

\newpage 

\printbibliography

@book{Cohen1993,
   author = {Henri Cohen},
   city = {Berlin, Heidelberg},
   doi = {10.1007/978-3-662-02945-9},
   isbn = {978-3-642-08142-2},
   publisher = {Springer Berlin Heidelberg},
   title = {A Course in Computational Algebraic Number Theory},
   volume = {138},
   year = {1993},
}

@book{Cohen2000,
   author = {Henri Cohen},
   city = {New York, NY},
   doi = {10.1007/978-1-4419-8489-0},
   isbn = {978-0-387-98727-9},
   publisher = {Springer New York},
   title = {Advanced Topics in Computional Number Theory},
   volume = {193},
   year = {2000},
}

@article{Silvester2000,
   author = {John R. Silvester},
   doi = {10.2307/3620776},
   issn = {0025-5572},
   issue = {501},
   journal = {The Mathematical Gazette},
   month = {11},
   pages = {460-467},
   title = {Determinants of block matrices},
   volume = {84},
   year = {2000},
}

@book{Kirschmer2016,
    place={RWTH Aachen University},
    title={Definite quadratic and hermitian forms with small class number (Habilitation)},
    url={https://www.math.rwth-aachen.de/~Markus.Kirschmer/papers/herm.pdf},
    author={Kirschmer, Markus}, 
    year={2016} }

@inproceedings{Nemo,
  author = {Fieker, Claus and Hart, William and Hofmann, Tommy and Johansson, Fredrik},
  title = {Nemo/Hecke: Computer Algebra and Number Theory Packages for the Julia Programming Language},
  booktitle = {Proceedings of the 2017 ACM on International Symposium on Symbolic and Algebraic Computation},
  series = {ISSAC '17},
  year = {2017},
  pages = {157--164},
  numpages = {8},
  doi = {10.1145/3087604.3087611},
  publisher = {ACM},
  address = {New York, NY, USA},
}

@misc{k2024,
      title={Stabilizer operators and Barnes-Wall lattices}, 
      author={Vadym Kliuhcnikov and Sebastian Sch{\"o}nnenbeck},
      eprint={2404.17677},
      doi = {10.48550/arXiv.2404.17677},
      year={2024},
}

@misc{Kirschmer2016code,
    title={Quadratic and hermitian lattices over number fields},
    author={Kirschmer, Markus}, 
    year={2016},
    url={https://www.math.rwth-aachen.de/~Markus.Kirschmer/magma/lat.html}
}

@article{Gosset2014,
author = {Gosset, David and Kliuchnikov, Vadym and Mosca, Michele and Russo, Vincent},
title = {An Algorithm for the T-Count},
year = {2014},
issue_date = {November 2014},
publisher = {Rinton Press, Incorporated},
address = {Paramus, NJ},
volume = {14},
number = {15–16},
issn = {1533-7146},
journal = {Quantum Info. Comput.},
month = {11},
pages = {1261–1276},
numpages = {16},
eprint={1308.4134},
archivePrefix={arXiv},
doi={10.48550/arXiv.1308.4134}
}

@article{Conti90,
author = {Conti, Pasqualina},
title = {Hermite canonical form and Smith canonical form of a matrix over a principal ideal domain},
year = {1990},
issue_date = {Jul. 1990},
publisher = {Association for Computing Machinery},
address = {New York, NY, USA},
volume = {24},
number = {3},
issn = {0163-5824},
url = {https://doi.org/10.1145/101104.101105},
doi = {10.1145/101104.101105},
journal = {SIGSAM Bull.},
month = {7},
pages = {8–16},
numpages = {9}
}

@article{Fieker2014,
   author = {Claus Fieker and Tommy Hofmann},
   doi = {10.1112/S1461157014000291},
   issn = {1461-1570},
   issue = {A},
   journal = {LMS Journal of Computation and Mathematics},
   month = {8},
   pages = {349-365},
   title = {Computing in quotients of rings of integers},
   volume = {17},
   year = {2014},
}

@inbook{Flint,
   author = {William B. Hart},
   doi = {10.1007/978-3-642-15582-6_18},
   pages = {88-91},
   title = {Fast Library for Number Theory: An Introduction},
   year = {2010},
}

@article{Pernet2010,
   author = {Clément Pernet and William Stein},
   doi = {10.1016/j.jnt.2010.01.017},
   issn = {0022314X},
   issue = {7},
   journal = {Journal of Number Theory},
   month = {7},
   pages = {1675-1683},
   title = {Fast computation of Hermite normal forms of random integer matrices},
   volume = {130},
   year = {2010},
}

@book{jacobson2009basic,
  title={Basic Algebra I: Second Edition},
  author={Jacobson, N.},
  isbn={9780486471891},
  lccn={2009006506},
  series={Basic Algebra},
  year={2009},
  publisher={Dover Publications}
}

@book{Edelkamp2010,
author = {Edelkamp, Stefan and Schroedl, Stefan and Koenig, Sven},
title = {Heuristic Search: Theory and Applications},
year = {2010},
isbn = {0123725127},
publisher = {Morgan Kaufmann Publishers Inc.},
address = {San Francisco, CA, USA}
}

@article{Lubeck2002,
   author = {Frank Lubeck},
   doi = {10.1006/jsco.2000.0430},
   issn = {07477171},
   issue = {1},
   journal = {Journal of Symbolic Computation},
   month = {1},
   pages = {57-65},
   title = {On the Computation of Elementary Divisors of Integer Matrices},
   volume = {33},
   year = {2002},
}

@article {AbramenkoNebeBuildings,
    AUTHOR = {Abramenko, Peter and Nebe, Gabriele},
     TITLE = {Lattice chain models for affine buildings of classical type},
   JOURNAL = {Math. Ann.},
  FJOURNAL = {Mathematische Annalen},
    VOLUME = {322},
      YEAR = {2002},
    NUMBER = {3},
     PAGES = {537--562},
      ISSN = {0025-5831},
     CODEN = {MAANA},
   MRCLASS = {20E42 (11E57 20G25)},
  MRNUMBER = {1895706},
MRREVIEWER = {Dmitrii V. Pasechnik},
       DOI = {10.1007/s002080200004},
       URL = {http://dx.doi.org/10.1007/s002080200004},
}

@book{Narkiewicz2004,
   author = {Władysław Narkiewicz},
   city = {Berlin, Heidelberg},
   doi = {10.1007/978-3-662-07001-7},
   isbn = {978-3-642-06010-6},
   publisher = {Springer Berlin Heidelberg},
   title = {Elementary and Analytic Theory of Algebraic Numbers},
   year = {2004},
}

@book{OMearaQuad,
   author = {O. Timothy O’Meara},
   city = {Berlin, Heidelberg},
   doi = {10.1007/978-3-642-62031-7},
   isbn = {978-3-540-66564-9},
   publisher = {Springer Berlin Heidelberg},
   title = {Introduction to Quadratic Forms},
   year = {1963},
}

@inbook{Voight2013,
   author = {John Voight},
   doi = {10.1007/978-1-4614-7488-3_10},
   pages = {255-298},
   title = {Identifying the Matrix Ring: Algorithms for Quaternion Algebras and Quadratic Forms},
   year = {2013},
}

@article {Magma,
    AUTHOR = {Bosma, Wieb and Cannon, John and Playoust, Catherine},
     TITLE = {The {M}agma algebra system. {I}. {T}he user language},
      NOTE = {Computational algebra and number theory (London, 1993)},
   JOURNAL = {J. Symbolic Comput.},
  FJOURNAL = {Journal of Symbolic Computation},
    VOLUME = {24},
      YEAR = {1997},
    NUMBER = {3-4},
     PAGES = {235--265},
      ISSN = {0747-7171},
   MRCLASS = {68Q40},
  MRNUMBER = {MR1484478},
       DOI = {10.1006/jsco.1996.0125},
       URL = {http://dx.doi.org/10.1006/jsco.1996.0125},
}

@article{Plesken1997,
author = {Plesken, W. and Souvignier, B.},
title = {Computing isometries of lattices},
year = {1997},
issue_date = {Sept./Oct. 1997},
publisher = {Academic Press, Inc.},
address = {USA},
volume = {24},
number = {3–4},
issn = {0747-7171},
url = {https://doi.org/10.1006/jsco.1996.0130},
doi = {10.1006/jsco.1996.0130},
journal = {J. Symb. Comput.},
month = {10},
pages = {327–334},
numpages = {8}
}

@misc{gidney2021cccz,
      title={A CCCZ gate performed with 6 T gates}, 
      author={Craig Gidney and N. Cody Jones},
      year={2021},
      eprint={2106.11513},
      archivePrefix={arXiv},
      primaryClass={quant-ph},
      doi={10.48550/arXiv.2106.11513}
}

@article{Jones2013,
   title={Low-overhead constructions for the fault-tolerant Toffoli gate},
   volume={87},
   ISSN={1094-1622},
   url={http://dx.doi.org/10.1103/PhysRevA.87.022328},
   DOI={10.1103/physreva.87.022328},
   number={2},
   journal={Physical Review A},
   publisher={American Physical Society (APS)},
   author={Jones, Cody},
   year={2013},
   month={02}
}

@misc{amy2024exact,
      title={Exact Synthesis of Multiqubit Clifford-Cyclotomic Circuits}, 
      author={Matthew Amy and Andrew N. Glaudell and Shaun Kelso and William Maxwell and Samuel S. Mendelson and Neil J. Ross},
      year={2024},
      eprint={2311.07741},
      archivePrefix={arXiv},
      primaryClass={quant-ph}
}

@article{Glaudell2021,
   author = {Andrew N. Glaudell and Neil J. Ross and Jacob M. Taylor},
   doi = {10.1038/s41534-021-00424-z},
   issn = {2056-6387},
   issue = {1},
   journal = {npj Quantum Information},
   month = {6},
   pages = {103},
   title = {Optimal two-qubit circuits for universal fault-tolerant quantum computation},
   volume = {7},
   year = {2021},
}

@article{Mukhopadhyay2024,
   author = {Priyanka Mukhopadhyay},
   month = {1},
   title = {Synthesizing Toffoli-optimal quantum circuits for arbitrary multi-qubit unitaries},
   year = {2024},
}

@article{Mukhopadhyay2024b,
   author = {Priyanka Mukhopadhyay},
   doi = {10.1103/PhysRevA.109.052619},
   issn = {2469-9926},
   issue = {5},
   journal = {Physical Review A},
   month = {5},
   pages = {052619},
   title = {Synthesis of V-count-optimal quantum circuits for multiqubit unitaries},
   volume = {109},
   year = {2024},
}

@article{Gheorghiu2021,
   author = {Vlad Gheorghiu and Michele Mosca and Priyanka Mukhopadhyay},
   doi = {10.1038/s41534-022-00651-y},
   month = {10},
   title = {T-count and T-depth of any multi-qubit unitary},
   year = {2021},
}

@article{Gheorghiu2021b,
   author = {Vlad Gheorghiu and Michele Mosca and Priyanka Mukhopadhyay},
   doi = {10.1038/s41534-022-00624-1},
   month = {1},
   title = {A (quasi-)polynomial time heuristic algorithm for synthesizing T-depth optimal circuits},
   year = {2021},
}

@article{Mosca2020,
   author = {Michele Mosca and Priyanka Mukhopadhyay},
   doi = {10.1088/2058-9565/ac2d3a},
   month = {6},
   title = {A polynomial time and space heuristic algorithm for T-count},
   year = {2020},
}

@article{Kalra2024,
   author = {Amolak Ratan Kalra and Manimugdha Saikia and Dinesh Valluri and Sam Winnick and Jon Yard},
   month = {5},
   title = {Multi-qutrit exact synthesis},
   year = {2024},
}

@article{Kalra2023,
   author = {Amolak Ratan Kalra and Dinesh Valluri and Michele Mosca},
   month = {11},
   title = {Synthesis and Arithmetic of Single Qutrit Circuits},
   year = {2023},
}

@article{Jain2020,
   author = {Akalank Jain and Amolak Ratan Kalra and Shiroman Prakash},
   month = {11},
   title = {A Normal Form for Single-Qudit Clifford+$T$ Operators},
   year = {2020},
}

@article{Bian2023,
   author = {Xiaoning Bian and Peter Selinger},
   doi = {10.4204/EPTCS.384.7},
   month = {6},
   title = {Generators and Relations for 3-Qubit Clifford+CS Operators},
   year = {2023},
}

@article{Bian2022,
   author = {Xiaoning Bian and Peter Selinger},
   doi = {10.4204/EPTCS.394.2},
   month = {4},
   title = {Generators and Relations for 2-Qubit Clifford+T Operators},
   year = {2022},
}

@article{Makary2021,
   author = {Justin Makary and Neil J. Ross and Peter Selinger},
   doi = {10.4204/EPTCS.343.2},
   month = {9},
   title = {Generators and Relations for Real Stabilizer Operators},
   year = {2021},
}

@article{Li2021,
   author = {Sarah Meng Li and Neil J. Ross and Peter Selinger},
   doi = {10.4204/EPTCS.343.11},
   month = {6},
   title = {Generators and Relations for the Group On(Z[1/2])},
   year = {2021},
}

@article{Bian2021,
   author = {Xiaoning Bian and Peter Selinger},
   doi = {10.4204/EPTCS.343.8},
   month = {5},
   title = {Generators and Relations for Un(Z[1/2,i])},
   year = {2021},
}

@article{Selinger2013,
   author = {Peter Selinger},
   doi = {10.2168/LMCS-11(2:10)2015},
   month = {10},
   title = {Generators and relations for n-qubit Clifford operators},
   year = {2013},
}

@article{Giles2012,
   author = {Brett Giles and Peter Selinger},
   doi = {10.1103/PhysRevA.87.032332},
   month = {12},
   title = {Exact synthesis of multiqubit Clifford+T circuits},
   year = {2012},
}

@article{Glaudell2024,
   author = {Andrew N. Glaudell and Neil J. Ross and John van de Wetering and Lia Yeh},
   month = {5},
   title = {Exact Synthesis of Multiqutrit Clifford-Cyclotomic Circuits},
   year = {2024},
}

@article{Amy2023,
   author = {Matthew Amy and Andrew N. Glaudell and Shaun Kelso and William Maxwell and Samuel S. Mendelson and Neil J. Ross},
   month = {11},
   title = {Exact Synthesis of Multiqubit Clifford-Cyclotomic Circuits},
   year = {2023},
}

@article{Amy2023b,
   author = {Matthew Amy and Andrew N. Glaudell and Sarah Meng Li and Neil J. Ross},
   month = {5},
   title = {Improved Synthesis of Toffoli-Hadamard Circuits},
   year = {2023},
}

@article{Glaudell2022,
   author = {Andrew Glaudell and Neil J. Ross and John van de Wetering and Lia Yeh},
   doi = {10.4230/LIPIcs.TQC.2022.12},
   month = {2},
   title = {Qutrit metaplectic gates are a subset of Clifford+T},
   year = {2022},
}

@article{Amy2019,
   author = {Matthew Amy and Andrew N. Glaudell and Neil J. Ross},
   doi = {10.22331/q-2020-04-06-252},
   month = {8},
   title = {Number-Theoretic Characterizations of Some Restricted Clifford+T Circuits},
   year = {2019},
}

@article{Glaudell2018,
   author = {Andrew N. Glaudell and Neil J. Ross and Jacob M. Taylor},
   doi = {10.1016/j.aop.2019.04.001},
   month = {3},
   title = {Canonical forms for single-qutrit Clifford+T operators},
   year = {2018},
}

@article{Evra2024,
   author = {Shai Evra and Ori Parzanchevski},
   month = {1},
   title = {Arithmeticity and covering rate of the $9$-cyclotomic Clifford+$\mathcal\{D\}$ gates in $PU(3)$},
   year = {2024},
}

@article{Evra2018,
   author = {Shai Evra and Ori Parzanchevski},
   doi = {10.1007/s00039-022-00593-9},
   month = {10},
   title = {Ramanujan complexes and Golden Gates in PU(3)},
   year = {2018},
}

@article{Parzanchevski2017,
   author = {Ori Parzanchevski and Peter Sarnak},
   doi = {10.1016/j.aim.2017.06.022},
   month = {4},
   title = {Super-Golden-Gates for PU(2)},
   year = {2017},
}

@article{Kliuchnikov2015,
   author = {Vadym Kliuchnikov and Jon Yard},
   month = {4},
   title = {A framework for exact synthesis},
   year = {2015},
}

@article{Kliuchnikov2013,
   author = {Vadym Kliuchnikov and Alex Bocharov and Krysta M. Svore},
   doi = {10.1103/PhysRevLett.112.140504},
   month = {10},
   title = {Asymptotically Optimal Topological Quantum Compiling},
   year = {2013},
}

@article{Kliuchnikov2012,
   author = {Vadym Kliuchnikov and Dmitri Maslov and Michele Mosca},
   month = {6},
   title = {Fast and efficient exact synthesis of single qubit unitaries generated by Clifford and T gates},
   year = {2012},
}

@article{Forest2015,
   author = {Simon Forest and David Gosset and Vadym Kliuchnikov and David McKinnon},
   doi = {10.1063/1.4927100},
   month = {1},
   title = {Exact synthesis of single-qubit unitaries over Clifford-cyclotomic gate sets},
   year = {2015},
}

\newpage

\end{document}